\documentclass[11pt,letterpaper]{article}
\usepackage[backref, colorlinks,citecolor=blue,linkcolor=magenta,bookmarks=true]{hyperref}

\usepackage{verbatim}
\usepackage{amsfonts}
\usepackage{amsmath,amssymb,amsthm}
\usepackage{multirow}
\usepackage{url}
\usepackage{setspace}
\usepackage{times}
\usepackage{fullpage}
\usepackage{epsfig}
\usepackage[usenames,dvipsnames]{xcolor}
\usepackage{graphicx}
\usepackage{times}
\usepackage{enumerate}
\usepackage{enumitem}
\usepackage[sc]{mathpazo}
\usepackage[T1]{fontenc}
\usepackage{algorithm}
\usepackage[noend]{algpseudocode}
\usepackage{relsize}

\newtheorem{theorem}{Theorem}[section]
\newtheorem{lemma}[theorem]{Lemma}

\newtheorem{claim}[theorem]{Claim}

\newtheorem{definition}[theorem]{Definition}

\def\2DF{\textsc{2D-Linear-FIXP}}

\def\rjunta{{\cal JUNTA}}
\def\bb{\mathbf{b}}   
 \def\calC{\mathcal{C}}  \def\11{\mathbf{1}}

 \def\calD{\mathcal{D}} \def\YES{\mathcal{YES}}
 \def\NO{\mathcal{NO}}
\def\eps{\epsilon}
\def\poly{\mathrm{poly}}
\def\bx{\mathbold{x}}
\def\boldf{\mathbold{f}} 
\def\bg{\mathbold{g}}
\def\bj{\mathbold{j}}
\def\bh{\mathbold{h}}
\def\bS{\mathbold{S}}
\def\bJ{\mathbold{J}}
\def\Bin{\mathrm{Bin}}

\def\bphi{\boldsymbol{\phi}}
\def\by{\mathbold{y}}

\def\bQ{\mathbold{Q}}
\def\bY{\mathbold{Y}}

\def\bZ{\mathbold{Z}}

\def\bz{\mathbold{z}}

\def\dist{\mathsf{dist}}
\def\D{{\cal D}}
\def\uhr{\hspace{-0.15cm}\upharpoonright}
\newcommand{\ignore}[1]{{}}

\def\colorful{0}
\ifnum\colorful=1

\newcommand{\red}[1]{{\color{red} {#1}}}

\fi
\ifnum\colorful=0

\newcommand{\red}[1]{{{#1}}}

\fi

\spacing{1.1}

\title{Distribution-free Junta Testing\footnote{
This work
is supported by NSF CCF-1703925, NSF CCF-1420349, and NSF CCF-1563155.}}
\author{
Xi Chen\footnote{Columbia University. Email: \texttt{xichen@cs.columbia.edu}.}\and
Zhengyang Liu\footnote{Shanghai Jiao Tong University. Email: \texttt{lzy5118@sjtu.edu.cn}.} \and
Rocco A. Servedio\footnote{Columbia University. Email: \texttt{rocco@cs.columbia.edu}.}\and
Ying Sheng\footnote{Columbia University. Email: \texttt{ys2982@columbia.edu}.}\and
Jinyu Xie\footnote{Columbia University. Email: \texttt{jinyu@cs.columbia.edu}.}}
\begin{document}

\begin{titlepage}

\maketitle

\begin{abstract}

We study the problem of testing whether an unknown $n$-variable Boolean function is a $k$-junta in the
\emph{distribution-free} property testing model, where the distance between functions is measured with respect to an arbitrary and unknown probability distribution over $\{0,1\}^n$.  Our first main result is that distribution-free $k$-junta testing can be performed, with one-sided error, by an adaptive algorithm that uses $\tilde{O}(k^2)/\eps$ queries (independent of $n$).  Complementing this, our second main result is a lower bound showing that any \emph{non-adaptive} distribution-free $k$-junta testing algorithm must make $\Omega(2^{k/3})$ queries even to test to accuracy $\eps=1/3$.  
These bounds establish that while the optimal query complexity of non-adaptive $k$-junta testing is $2^{\Theta(k)}$, for adaptive testing it is $\poly(k)$, and thus show that adaptivity provides an exponential improvement in the distribution-free query complexity of testing juntas.

\ignore{These bounds settle the query complexity of both adaptive and non-adaptive distribution-free junta testing to within a polynomial factor\footnote{Jinyu: the upper bound in non-adaptive setting is $O(2^k)$ so this is not exactly within a polynomial factor. I get that we want to say something similar to this though.}, and show that adaptivity provides an exponential improvement in the distribution-free query complexity of testing juntas.}

\ignore{  We prove that adaptive testing is more powerful than
  non-adaptive testing for \(k\)-junta in {\em distribution-free}
  model. That is, we give a one-sided error algorithm for testing
  \(k\)-junta with \(\tilde{O}(k^2)/\epsilon\) queries adaptively, and
  also show an
  \(\Omega(2^{k/3})\) lower bound for testing \(k\)-junta non-adaptively.
  }

\end{abstract}

\thispagestyle{empty}
\end{titlepage}

\newpage


\def\bR{\mathbf{R}}
\def\bP{\mathbf{P}}
\def\bQ{\mathbf{Q}}
\def\bD{\boldsymbol{\mathcal{D}}}

\section{Introduction}
\label{sec:introduction}

Property testing of Boolean functions was first considered in the seminal works of Blum, Luby and Rubinfeld \cite{BLR93} and Rubinfeld and Sudan \cite{Rubinfeld1996} and has developed into a robust research area at the intersection of sub-linear algorithms and complexity theory.  Roughly speaking, a property tester for a class $\calC$ of functions from $\{0,1\}^n$ to $\{0,1\}$ is a randomized algorithm that is given some form of access to the (unknown) input Boolean function $f$, and must with high probability distinguish the case that $f \in \calC$ versus the case that $f$ is $\eps$-far from every function $g \in \calC$.  In the usual (uniform-distribution) property testing scenario, the testing algorithm may access $f$ by making black-box queries on inputs $x \in \{0,1\}^n$, and the distance between two functions $f$ and $g$ is measured with respect to the uniform distribution on $\{0,1\}^n$; the goal is to develop algorithms that make as few queries as possible.
Many different classes of Boolean functions~have been studied from this perspective, see \cite{BLR93,AKKLRtit,BKSSZ10,PRS02,DLM+:07,GGLRS,FLNRRS,CS13a,CST14,CDST15,KMS15,BB15,KS16,CS16,BMPR16,CWX17,CWX17focs,MORS:10,MORS:09random,BlaisBM11,DLM+:07,BlaisBM11,BlaisKane12,GOS+11} and other works referenced in the surveys \cite{Ron:08testlearn,Ron:10FNTTCS,PropertyTestingICS}.\ignore{(e.g. linear functions and low-degree polynomials over $GF(2)$ \cite{BLR93,AKKLRtit,BKSSZ10}, literals, conjunctions, $s$-term monotone and non-monotone DNFs \cite{PRS02,DLM+:07}, monotone and unate functions \cite{GGLRS,FLNRRS,CS13a,CST14,CDST15,KMS15,BB15,KS16,CS16,BMPR16,CWX17,CWX17focs},
various types of linear threshold functions \cite{MORS:10,MORS:09random,BlaisBM11}, size-$s$ decision trees and $s$-sparse $GF(2)$ polynomials and parities \cite{DLM+:07,BlaisBM11,BlaisKane12}, functions with sparse or low-degree Fourier spectrum \cite{GOS+11}, and more),}
Among these, the~class of \emph{$k$-juntas} --- Boolean functions that depend only on (an unknown set of)~at~most~$k$~of~their~$n$~input variables --- is one of the best-known and most intensively investigated such classes \cite{FKRSS03,ChocklerGutfreund:04,Blais08,Blaisstoc09,BGSMdW13,STW15}, with ongoing research on junta testing continuing right up to the present \cite{SSTWX17}.

The query complexity of junta testing in the uniform distribution framework is now well understood.
 Improving on $\poly(k)/\eps$-query algorithms given in \cite{FKRSS03} (which introduced the
junta testing problem), in~\cite{Blais08} Blais gave a non-adaptive algorithm that makes
\(\tilde{O}(k^{3/2})/\epsilon\) queries, and in \cite{Blaisstoc09} Blais gave an \(O(k\log k + k/\epsilon)\)-query adaptive algorithm.
On the lower bounds side, Fischer et al.~\cite{FKRSS03} initially gave an
\(\Omega(\sqrt{k})\) lower bound for non-adaptively testing $k$-juntas, which also implies an \(\Omega(\log k)\) lower bound for
adaptive testing. Chockler and Gutfreund improved the adaptive lower bound
to $\Omega(k)$ in~\cite{ChocklerGutfreund:04}, and very recently Chen et al.~\cite{SSTWX17} gave an $\tilde{\Omega}(k^{3/2})/\eps$ non-adaptive lower bound. Thus in both the adaptive and non-adaptive uniform distribution settings, the query complexity of $k$-junta testing has now been pinned down to within logarithmic factors.

\medskip \noindent {\bf Distribution-free property testing.}  This work studies the junta testing problem in the \emph{distribution-free} property testing model that was first introduced by Goldreich \emph{et al} in~\cite{GGR98}.  In this model the distance between Boolean functions is measured with respect to a distribution $\calD$ over $\{0,1\}^n$ which is arbitrary and unknown to the testing algorithm.  Since the distribution is unknown, in this model the testing algorithm is allowed (in addition to making black-box queries) to draw random labeled samples $(\bx,f(\bx))$ where each $\bx$ is independently distributed according to $\calD$.  The query complexity of an algorithm in this framework is the worst-case total number of black-box oracle calls plus random labeled samples that are used, across all possible distributions.  (It follows that distribution-free testing of a class $\calC$ requires at least as many queries as testing $\calC$ in the standard uniform-distribution model.)

Distribution-free property testing is in the spirit of similar distribution-free models in computational learning theory such as Valiant's celebrated PAC learning model \cite{Valiant:84}.  Such models are attractive because of their minimal assumptions; they are well motivated both because in~many natural settings the uniform distribution over $\{0,1\}^n$ may not be the best way to measure distances, and because they capture the notion of an algorithm dealing with an unknown and arbitrary environment (modeled here by the unknown and arbitrary distribution $\D$ over $\{0,1\}^n$ and the unknown and arbitrary Boolean function $f: \{0,1\}^n \to \{0,1\}$). Researchers have studied distribution-free testing of a number of Boolean function classes, including monotone functions, low-degree polynomials, dictators (1-juntas) and $k$-juntas  \cite{HalevyKushilevitz:07}, disjunctions and conjunctions (monotone and non-monotone), decision lists, and linear threshold functions \cite{GlasnerServedio:09toc,DolevRon:11,CX16}. Since depending on few variables is an appealingly flexible ``real-world'' property in comparison with more highly structured syntactically defined properties,\ignore{ such as being computed exactly by a conjunction, by a linear threshold function, etc.,} we feel that junta testing is a particularly natural task to study in the distribution-free model.

\medskip

\noindent {\bf Prior results on distribution-free junta testing.}  Given how thoroughly junta testing has been studied in the uniform distribution model, surprisingly little was known in the \mbox{distribution-free} setting.  The adaptive $\Omega(k)$ and non-adaptive $\tilde{\Omega}(k^{3/2})/\eps$ uniform-distribution lower bounds from \cite{ChocklerGutfreund:04,SSTWX17} mentioned earlier trivially extend to the distribution-free model, but no other lower bounds on distribution-free junta testing were known prior to this work.  On the positive side, Halevy and Kushilevitz showed in \cite{HalevyKushilevitz:07}  that any class $\calC$ that has (i) a one-sided error uniform-distribution testing algorithm and (ii) a self-corrector, has a one-sided error distribution-free~testing algorithm.  As $\poly(k)/\eps$-query one-sided junta testers were given already in \cite{FKRSS03}, and $k$-juntas have $O(2^{k})$-query self-correctors \cite{AlonWeinstein12}, this yields a one-sided non-adaptive distribution-free junta tester with query complexity $O(2^k/\eps)$.  No other results were known.

Thus, prior to this work there were major gaps in our understanding of distribution-free $k$-junta testing:  is the query complexity of this problem polynomial in $k$, exponential in $k$, or somewhere in between?  Does adaptivity confer an exponential advantage, a sub-exponential advantage, or no advantage at all?  Our results, described below, answer both these questions.

\subsection{Our results}  Our main positive result is a $\poly(k)/\eps$-query one-sided adaptive algorithm for distribution-free $k$-junta testing:

\begin{theorem}[Upper bound]\label{main}
  For any  \(\epsilon>0\), there is a one-sided distribution-free adaptive
  algorithm for $\eps$-testing \(k\)-juntas with \(\tilde{O}(k^2)/\epsilon\) queries.
\end{theorem}

Theorem~\ref{main} shows that $k$-juntas stand in interesting contrast with many other well-studied classes
  of Boolean functions in property testing such as conjunctions, decision lists, linear threshold functions, and monotone functions. For each of these classes distribution-free testing requires dramatically more queries than uniform-distribution testing: for the first three classes~the separation is $\poly(1/\eps)$ queries in the uniform setting \cite{PRS02,MORS:10} versus $n^{\Omega(1)}$ queries in the distribution-free setting \cite{GlasnerServedio:09toc,CX16}; for $n$-variable monotone functions $\poly(n)$ queries suffice in the uniform setting \cite{GGLRS,KMS15} whereas \cite{HalevyKushilevitz:07} shows that $2^{\Omega(n)}$ queries are required in the distribution-free setting.  In contrast, Theorem~\ref{main} shows that for $k$-juntas the query complexities of uniform-distribution and distribution-free testing are polynomially related (indeed, within at most a quadratic factor of each other).

Complementing the strong upper bound which Theorem~\ref{main} gives for adaptive testers, our main negative result is an $\Omega(2^{k/3})$-query lower bound for non-adaptive testers:

\begin{theorem}[Lower bound]\label{lowerbound}
For  
  $k \leq n/200$, any non-adaptive algorithm that distribution-free $\eps$-tests  \(k\)-juntas over $\{0,1\}^n$, for $\eps=1/3,$ must have query complexity $\Omega(2^{k/3}).$
\end{theorem}

Theorems~\ref{main} and~\ref{lowerbound} together show that adaptivity enables an exponential improvement in the distribution-free query complexity of testing juntas.  This is in sharp contrast with uniform-distribution junta testing, where the adaptive and non-adaptive query complexities are polynomially related (with an exponent of only 3/2).  To the best of our knowledge, this is the first example of a exponential separation between adaptive and nonadaptive
distribution-free testers.

\subsection{Ideas and techniques}

\noindent {\bf The algorithm.}  As a first step toward our $\tilde{O}(k^2)/\eps$-query algorithm, in Section~\ref{sec:warmup} we first present a simple one-sided adaptive algorithm, which we call {\bf SimpleDJunta}, that distribution-free tests $k$-juntas using $O({(k/ \eps)} + k \log n)$ queries.   {\bf SimpleDJunta} uses binary search and is an adaptation to the distribution-free setting of the $O({(k/ \eps)} + k \log n)$-query 
  uniform-distribution algorithm which is implicit in \cite{Blaisstoc09}.  
The algorithm maintains a set $I$ of \emph{relevant} variables: 
  a string $x\in \{0,1\}^n$ has been found for each $i\in I$ such that $f(x)\ne f(x^{(i)})$
  (we use $x^{(i)}$ to denote the string obtained by flipping the $i$-th bit of $x$),
  and the algorithm rejects only when $|I|$ becomes larger than $k$.
In each round, the algorithm samples a string $\bx\leftarrow \D$ and 
  a subset $\bR$ of $\overline{I}:=[n]\setminus I$ uniformly at random.
A simple lemma, Lemma~\ref{influence}, states that if $f$ is far from every $k$-junta with respect to $\calD$, then $f(\bx)\ne f(\bx^{(\bR)})$ with at least some moderately large probability as long as $|I|\le k$, where we use $\bx^{(\bR)}$ to
  denote the string obtained from $\bx$ by flipping every coordinate in $\bR$.
With such a pair $(\bx,\bx^{(\bR)})$ in hand, it is straightforward to find a 
  new relevant variable using binary search over coordinates in $\bR$ 
  (see Figure \ref{fig:binarysearch}), with at most $\log n$ additional queries.

In order to achieve a query complexity that is independent of $n$, clearly one must employ a more efficient approach than binary search over $\Omega(n)$ coordinates (since most likely the set~$\bR$~has size $\Omega(n)$ for 
the range of $k$ we are interested in). In the uniform-distribution setting this~is~accomplished in \cite{Blaisstoc09} by first randomly partitioning the variable space $[n]$ into $s=\poly(k/\eps)$ disjoint blocks $B_1,\ldots,B_s$ of variables and carrying out binary search over blocks (see Figure \ref{fig:blockbinarysearch})
rather than over individual coordinates; this reduces the cost of each binary search to $\log(k/\eps)$ rather than $\log n$. 
The algorithm maintains a set of \emph{relevant blocks}: two strings $x,y\in \{0,1\}^n$
  have been found for each such block $B$ which satisfy $f(x)\ne f(y)$ and $y=x^{(S)}$ with $S\subseteq B$,
  and the algorithm rejects when more than $k$ relevant blocks have been found.
In each round the algorithm samples two strings $\bx,\by$ uniformly at random conditioned on their agreeing with each other on the relevant blocks that have already been found in previous rounds; if $f(\bx)\ne f(\by),$ then the binary search over blocks is performed to
  find a new relevant block.
To establish the correctness of this approach \cite{Blaisstoc09} employs a detailed and technical analytic argument based on the influence of coordinates and the Efron-Stein orthogonal decomposition of functions over product spaces.  This machinery is well suited for dealing with product distributions, and indeed the analysis~of \cite{Blaisstoc09} goes through for any product distribution over $\{0,1\}^n$ (and even for more general finite domains and ranges).  However, it is far from clear how to extend this machinery to work for the completely unstructured distributions $\calD$ that must be handled in the distribution-free model.
\newpage 

Our main distribution-free junta testing algorithm, denoted {\bf MainDJunta}, 
  draws ideas from both {\bf SimpleDJunta} (mainly Lemma \ref{influence}) 
  and the uniform distribution tester of \cite{Blaisstoc09}.
To avoid~the $\log n$ cost, the algorithm
  carries out binary search over blocks rather than over individual coordinates,
  and maintains a set of disjoint relevant blocks $B_1,\ldots,B_\ell$, i.e., for each $B_j$
  a pair of strings $x^j$ and $y^j$ have been 
  found such that they agree with each other over $\overline{B_j}$ and satisfy $f(x^j)\ne f(y^j)$.
Let $w^j$ be the projection of $x^j$ (and $y^j$) over $\overline{B_j}$ and 
  let $g_j$ be the Boolean function over $\{0,1\}^{B_j}$ obtained from $f$ by setting
  variables in $\overline{B_j}$ to $w^j$.
For clarity we assume further that every function $g_j$ is very close to a \emph{literal} 
  (i.e. for some $\tau \in \{x_{i_j},\overline{x_{i_j}}\}$ we have
  $g_j(x)=\tau$ for all $x\in \{0,1\}^{B_j}$ for some $i_j\in B_j$) under the \emph{uniform} 
  distribution.
(To justify this assumption we note that if $g_j$ is far from every literal under the uniform
  distribution, then
  it is easy to split $B_j$ further into two relevant blocks using the uniform distribution
  algorithm of \cite{Blaisstoc09}.)
Let $I=\{i_j:j\in [\ell]\}$.
Even though the algorithm does not know $I$,
  there is indeed a way to draw uniformly random subsets $\bR$ of $\overline{I}$.
First  we draw a partition of $B_j$ into $\bP_j$ and $\bQ_j$ uniformly at random, for each $j$.
Since $g_j$ is close to a literal, it is not difficult to figure out 
  whether $\bP_j$ or $\bQ_j$ contains the hidden $i_j$, say it is $\bP_j$ for every $j$.
Then the union of all $\bQ_j$'s together with a uniformly random 
  subset of $\overline{B_1\cup \cdots \cup B_\ell}$,
  denoted by $\bR$,
  turns out to be a uniformly random subset of $\overline{I}$.
With $\bR$ in hand,  Lemma \ref{influence} implies that
  $f(\bx)\ne f(\bx^{(\bR)})$ with high probability when
 $\bx\leftarrow \D$, and when this happens, one can carry out binary search
  over blocks to increase the number of relevant blocks by one.
In Section \ref{sec:intuition} we explain the intuition behind the main algorithm in more detail.

\medskip
\noindent {\bf The lower bound.}
As we explain in Section~\ref{sec:preliminaries}, a $q$-query non-adaptive distribution-free tester is a randomized algorithm $A$ that works as follows. 
When $A$ is run on an input pair $(\phi,\D)$\footnote{For clarity, throughout our discussion of lower bounds we write $\phi$ to indicate a function which may be either a ``yes-function'' or a ``no-function'', $f$ to denote a ``yes-function'' and $g$ to denote a ``no-function.''} it is first 
  given the result $(\by^1,\phi(\by^1)),\dots,(\by^q,\phi(\by^q))$ of $q$ queries from the sampling oracle.
Based~on~them, it queries the black-box oracle $q$ times on strings $\bz^1,\dots,\bz^q.$  The $\bz^j$'s may depend on the random pairs $(\by^i,\phi(\by^i))$ received from the sampling oracle, but the $j$-th black-box query string $\bz^j$ may not depend on the responses $\phi(\bz^1),\dots,\phi(\bz^{j-1})$ to any of the $j-1$ earlier black-box queries.

As is standard in property testing lower bounds, our argument employs a distribution ${\cal YES}$ over yes-instances and a distribution ${\cal NO}$ over no-instances. \ignore{we employ Yao's principle, i.e.~we prove lower bounds for deterministic non-adaptive algorithms which receive inputs drawn one of two probability distributions, either ${\cal YES}$ or ${\cal NO}$.}Here ${\cal YES}$ is a distribution over (function, distribution) pairs $(\boldf,\bD )$ in which $\boldf$ is guaranteed to be a $k$-junta; ${\cal NO}$ is a distribution over pairs $(\bg,\bD)$ such that with probability  {$1-o(1)$}, $\bg$ is  {$1/3$}-far from every $k$-junta with respect to $\bD .$
To prove the desired lower bound against non-adaptive distribution-free testers, it suffices to show that for $q=2^{k/3}$, any deterministic non-adaptive algorithm $A$ as described above is roughly equally likely to accept whether it is run on an input drawn from ${\cal YES}$ or from ${\cal NO}.$

Our construction of the ${\cal YES}$ and ${\cal NO}$ distributions is essentially as follows.  In making~a draw either from ${\cal YES}$ or from ${\cal NO}$, first $m=\Theta(2^k \log n)$ strings are selected uniformly at~random from $\{0,1\}^n$ to form a set $\bS$, and the distribution $\bD$ in
  both ${\cal YES}$ and ${\cal NO}$ is set to be the uniform distribution over $\bS$.  Also in both ${\cal YES}$ and ${\cal NO}$, a ``background'' $k$-junta $\bh$ is selected uniformly at random by 
  first picking a set $\bJ$ of $k$ variables at random and then a random truth table for $\bh$ over the variables in $\bJ$.  We view the variables in $\bJ$ as partitioning $\{0,1\}^n$ into $2^k$ disjoint ``sections'' depending on how they are set.

In the case of a draw from ${\cal YES}$, the Boolean function $\boldf$ that goes with the above-described $\bD$ is simply the background junta $\boldf=\bh$.  In the case of a draw from ${\cal NO}$, the function $\bg$ that goes with $\bD$ is formed by modifying the background junta $\bh$ in the following way (roughly speaking;  see Section~\ref{sec:distributions} for precise details):  for each $z \in \bS$, we toss a fair coin
  $\bb(z)$ and set the value of all the strings in $z$'s section that lie within Hamming distance $0.4n$ from $z$ (including $z$ itself) to $\bb(z)$ (see Figure~\ref{fig:example-x}).  Note that the value of $\bg$ at each string in $\bS$ is a fair coin toss, which is completely independent of the background junta $\bh.$  Using the choice of $m$ it can be argued (see Section~\ref{sec:distributions}) that with high probability $\bg$ is $1/3$-far from every $k$-junta with respect to $\bD$ as $(\bg,\bD)\leftarrow {\cal NO}$.

The rough idea of why a pair $(\boldf,\bD) \leftarrow {\cal YES}$ is difficult for a ($q=2^{k/3}$)-query non-adaptive algorithm $A$ to distinguish from a pair $(\bg,\bD ) \leftarrow {\cal NO}$ is as follows.  Intuitively, in order~for~$A$ to distinguish the no-case from the yes-case, it must obtain two strings $x^1,x^2$ that belong to the same section but are labeled differently.  Since there are $2^k$ sections but $q$ is only $2^{k/3}$, by the birthday paradox it is very unlikely that $A$ obtains two such strings among the $q$ samples $\by^1,\dots,\by^q$ that it is given from the distribution $\bD$.  In fact, in both the yes- and no- cases, writing $(\bphi,\bD )$ to denote the (function, distribution) pair, the distribution of the $q$ pairs $(\by^1,\bphi(\by^1)),\dots,(\by^q,\bphi(\by^q))$ will be statistically very close to $(\bx^1,\bb_1),\dots,(\bx^q,\bb_q)$ where each pair $(\bx^j,\bb_j)$ is independently drawn uniformly from $\{0,1\}^n \times \{0,1\}$.  Intuitively, this translates into the examples $(\by^i,\bphi(\by^i))$ from the sampling oracle ``having no useful information'' about the set $\bJ$ of variables that the background junta depends on.

What about the $q$ strings $\bz^1,\dots,\bz^q$ that $A$ feeds to the black-box oracle?  It is also unlikely that any two elements of $\by^1,\dots,\by^q,\bz^1,\dots,\bz^q$ belong to the same section but are labeled differently.  Fix an $i \in [q]$; we give some intuition here as to why
it is very unlikely that there is any $j$ such that $\bz^i$ lies in the same section as $\by^j$ but has $f(\bz^i) \neq f(\by^j)$
\ignore{why $\bz^i$ is very unlikely to lie in the same section as any $\by^j$ }(via a union bound, the same intuition handles all $i \in [q]$).  Intuitively, since the random examples from the sampling oracle provide no useful information about the set $\bJ$ defining the background junta, the only thing that $A$ can do in selecting $\bz^i$ is to choose how far it lies, in terms of Hamming distance, from the points in $\by^1,\dots,\by^q$ (which, recall, are uniform random).  Fix $j \in [q]$:  if $\bz^i$ is within Hamming distance $0.4n$ from $\by^j$, then even if $\bz^i$ lies in the same section as $\by^j$ it will be labeled the same way as $\by^j$ whether we are in the yes- case or the no- case.  On the other hand, if $\bz^i$ is farther than $0.4n$ in Hamming distance from $\by^j$, then it is overwhelmingly likely that $\bz^i$ will lie in a different section from $\by^j$ (since it is very unlikely that all $0.4n$ of the flipped coordinates avoid the $k$-element set $\bJ$).  We prove Theorem~\ref{lowerbound} in Section~\ref{sec:lower} via a formal argument that proceeds somewhat differently from but is informed by the above intuitions.

\ignore{
The first thing a non-adaptive algorithm does is make $q=2^{k/3}$ calls to the sampling oracle to receive $(\bz^1,\bphi(\bz^1)),\dots,(\bz^q,\bphi(\bz^q))$ where $\bphi$ is either $\boldf$ or $\bg$ (depending on whether the algorithm is interacting with a draw from ${\cal YES}$ or from ${\cal NO}$) and each $\bz^i$ is independently drawn from $\D_{\bphi}$.  A birthday paradox argument easily shows that with high probability the distribution of these $q$ pairs is identical in either case.  Next, the algorithm fixes a sequence of $q$ query strings $z'^1,z'^q$; these are chosen deterministically based on  $(\bz^1,\bphi(\bz^1)),\dots,(\bz^q,\bphi(\bz^q))$. \red{fix/finish}
}


\medskip

\noindent {\bf Organization.}
In Section~\ref{sec:preliminaries} we define the distribution-free testing
model and introduce some useful notation. In Section \ref{sec:warmup}
  we present {\bf SimpleDJunta} and prove Lemma \ref{influence} in its analysis.
In Section~\ref{sec:our-algorithm} we present our main algorithm {\bf MainDJunta} and 
  prove Theorem~\ref{main}, and in
Section~\ref{sec:lower} we prove Theorem~\ref{lowerbound}.


\def\diff{\mathsf{diff}}
\def\bR{\mathbf{R}}

\section{Preliminaries}
\label{sec:preliminaries}

\noindent {\bf Notation.}
We use $[n]$ to denote $\{1,\ldots,n\}$.
We use $f$ and $g$ to denote
Boolean functions, which are maps from $\{0,1\}^n$ to $\{0,1\}$ for some positive integer $n$. 
We use the calligraphic font (e.g.,~$\calD$~and ${\cal NO}$) to denote probability distributions,   boldface letters such as $\bx$ to denote random variables,
  and write ``$\bx \leftarrow \calD$'' to indicate that $\bx$ is a random variable drawn from
  a distribution $\calD.$
We write $\bx\leftarrow \{0,1\}^n$ to denote that $\bx$ is a string drawn uniformly at random. 
Given $S\subseteq [n]$, we also write $\bR\leftarrow S$ to indicate that 
  $\bR$ is a subset of $S$ drawn uniformly at random, i.e., each index $i\in S$ is
  included in $\bR$ independently with probability $1/2$.

Given a subset $B \subseteq [n]$, we use $\overline{B}$ to denote its compliment with respect to $[n]$,
  and $\{0,1\}^B$ to denote the set of all binary strings of length 
  $|B|$ with coordinates indexed by $i\in B$.
Given an $x \in \{0,1\}^n$ and a $B \subseteq [n]$, we write $x_B\in \{0,1\}^B$ to denote the
  projection of $x$ over coordinates in $B$
  and $x^{(B)}\in \{0,1\}^n$ to denote the string obtained from $x$ by flipping 
  coordinates in $B$.
Given $x\in \{0,1\}^B$ and $y\in \{0,1\}^{\overline{B}}$, we write
  $x\circ y\in \{0,1\}^n$ to denote their concatenation, a string that agrees with $x$ over 
  coordinates in $B$ and agrees with $y$ over $\overline{B}$.
(As an example of the notation, given $x,y\in \{0,1\}^n$ and $B\subseteq [n]$,
  $x_B\circ y_{\overline{B}}$ denotes the string that agrees with $x$ over $B$ and with
  $y$ over $\overline{B}$.)
Given $x,y \in \{0,1\}^n$, we write $d(x,y)$ to denote the Hamming distance between $x$ and $y$,
  and $\diff(x,y)\subseteq [n]$ to denote the set of coordinates $i\in [n]$ with $x_i\ne y_i$.

Given \(f,g:\{0,1\}^n\rightarrow\{0,1\}\) and a probability distribution \(\D\) over \(\{0,1\}^n\), we write
\[
  \mathsf{dist}_{\D}(f,g) :=\Pr_{\bz \leftarrow \D}\big[f(\bz) \neq g(\bz)\big]
\]
to denote the \emph{distance} between \(f\) and \(g\) with respect to  \(\D\).
Given a class \(\mathfrak{C}\) of Boolean functions,  
\[
  \dist_\D(f,\mathfrak{C}) :=
  \min_{g\in\mathfrak{C}}\big(\dist_{\D}(f, g)\big) 
\]
  denotes the \emph{distance}
between $f$ and $\mathfrak{C}$ with respect to $\D$, where the minimum is taken over 
  $g$ with the same number of variables as $f$.
We say \(f\) is {\em \(\epsilon\)-far from $\mathfrak{C}$ with respect to $\D$} if 
$\dist_{\D}(f,\mathfrak{C})\geq\epsilon$.

We will often work with restrictions of Boolean functions.  Given
$f: \{0,1\}^n \to \{0,1\}$, $S \subseteq [n]$ and a string $z\in \{0,1\}^B$,
  the \emph{restriction of $f$ over $B$ by $z$}, denoted by $f\uhr_z$, is the Boolean function
$g: \{0,1\}^{\overline{B}} \to \{0,1\}$ defined by
$g(x) = f(x\circ z)$ for all $x\in \{0,1\}^{\overline{B}}$.  

\medskip
\noindent {\bf Distribution-free property testing.}
Now we can define distribution-free property testing:

\begin{definition}
We say an algorithm $A$ has \emph{oracle access} to a pair $(f,\D)$, where $f:\{0,1\}^n\rightarrow \{0,1\}$~is an unknown Boolean function and $\D$ is an unknown probability distribution over $\{0,1\}^n$,
  if it can (1) access $f$ via a black-box oracle that returns $f(x)$ when a string $x\in \{0,1\}^n$
  is queried, and (2) access $\D$ via a sampling oracle that, upon each request, returns a pair $(\bx,f(\bx))$
  where $\bx\leftarrow \D$ independently. 

Let \(\mathfrak{C}\) be a class of Boolean functions.  
A {\em distribution-free testing algorithm \(A\) for \(\mathfrak{C}\)} 
  is a randomized algorithm that, given as input a distance parameter $\eps>0$
  and oracle access to a pair $(f,\D)$, 
  accepts with probability at least $2/3$ if $f\in \mathfrak{C}$ and
  rejects with probability at least $2/3$ if $f$ is $\eps$-far from $\mathfrak{C}$ with respect to $\D$. 
We say $A$ is \emph{one-sided} if it always accepts when $f\in \mathfrak{C}$.
The \emph{query complexity} of a distribution-free testing algorithm  
  is the number of queries made on $f$ plus the number of samples drawn from $\D$.
\end{definition}

One may assume 
  without loss of generality that a distribution-free 
  testing algorithm consists of two phases: 
In the first phase, the algorithm draws a certain number of sample pairs $(\bx,f(\bx))$ from $\D$;~in~the second phase, it makes black-box queries to $f$. 
In general, a query $x\in \{0,1\}^n$~made by the algorithm in the second phase may
  depend on sample pairs it receives in the first phase
  (e.g. it can choose to query a string that is close to a sample received in the first phase) 
  and~results of queries to $f$ made so far.
In Section~\ref{sec:lower} we will prove lower bounds on 
  \emph{non-adaptive} distribution-free testing algorithms.  
An algorithm is said to be non-adaptive if its black-box queries 
  made in the second phase do not depend on results of previous black-box queries,
  i.e., all queries during the second phase can be made in a single batch (though we emphasize that
  they may depend on samples the algorithm received in the first phase).
  
  

\medskip \noindent {\bf Juntas and literals.}
We study the distribution-free testing of the class of $k$-juntas.  Recall that a Boolean function $f$ is a
$k$-junta if it depends on at most $k$ variables. More precisely, $f$ is a $k$-junta if there exists a list $1 \leq i_1 < \cdots < i_{k} \leq n$ of $k$ indices and a Boolean function $g: \{0,1\}^{k} \to \{0,1\}$ over $k$
  variables such that $f(x_1,\dots,x_n) = g(x_{i_1},\dots,x_{i_{k}})$ for all $x \in \{0,1\}^n$.


\ignore{\begin{definition}[\(k\)-junta]
  A boolean function \(f:\{0,1\}^n\rightarrow\{0,1\}\) is {\em
    \(k\)-junta} if the value of \(f\) depends on at most \(k\)
  variables.
\end{definition}
}
We say that a Boolean function \(f\) is a \emph{literal} if $f$ depends on exactly one variable, i.e. for some
$i \in [n]$, we have that either $f(x) = x_i$ for all $x$ or $f(x)= \overline{x_i}$ for all $x$.  Note that the two constant (all-$1$ and all-$0$) functions are one-juntas but are \emph{not} literals.

We often use the term ``\emph{block}'' to refer to a \emph{nonempty} subset of $[n]$, 
  which should be interpreted as a nonempty subset of the $n$ variables of a Boolean function
  $f:\{0,1\}^n\rightarrow \{0,1\}$.
The following definition of distinguishing pairs and relevant blocks will 
  be heavily used in our algorithms.
  
\ignore{
Through the paper, we use the notation {\em block} to denote a subset
of \([n]\). We say {\em random partition} of a set \(S\) into \(s\)
blocks if each elements in \(S\) is assigned to each of \(s\) blocks
uniformly and independently at random.
}

\begin{definition}[Distinguishing pairs and relevant blocks]
Given \(x,y\in\{0,1\}^n\) and a block $B \subseteq [n]$, we say that \((x,y)\) is a {\em distinguishing
    pair for $B$} if
  \(x_{\overline{B}} = y_{\overline{B}}\) and \(f(x) \neq f(y)\).
We say $B$ is a \emph{relevant} block of $f$ if such a distinguishing pair exists for $B$
  (or equivalently, the influence of $B$ in $f$ is positive).

When $B=\{i\}$ is a relevant block we simply say that the $i$-th variable is relevant to $f$.  
\end{definition}

As will become clear later, all our algorithms reject a function $f$ only when 
  they have found $k+1$ pairwise disjoint blocks $B_1,\ldots,B_{k+1}$ and a distinguishing pair
  for each $B_i$.
When this occurs, it means that $B_1,\ldots,B_{k+1}$ are pairwise disjoint relevant blocks of $f$,
  which implies that $f$ cannot be a $k$-junta. 
As a result, our algorithms are one-sided.
To prove their correctness, it suffices to show that they
  reject with probability at least $2/3$ when $f$ is $\eps$-far from $k$-juntas with respect to~$\D$.

For the standard property testing model under the uniform distribution, Blais \cite{Blaisstoc09}
  obtained a nearly optimal algorithm:

\begin{theorem}[\cite{Blaisstoc09}]\label{blaistheorem}
  There exists a one-sided, \(O((k/\epsilon) + k\log k)\)-query algorithm \emph{\textbf{UniformJunta}}$(f,k,\eps)$
  that rejects $f$ with probability at least $2/3$
  when it is $\eps$-far from $k$-juntas under the uniform distribution.
Moreover, it rejects only when it has found $k+1$ pairwise disjoint blocks and
  a distinguishing pair of $f$~for each of them. 
\end{theorem}
\noindent {\bf Binary Search.}
The standard binary search procedure (see Figure \ref{fig:binarysearch}) takes as input two strings 
  $\smash{x,y\in \{0,1\}^n}$
  with $\smash{f(x)\ne f(y)}$, makes $\smash{O(\log n)}$ queries on $f$, and returns 
  a pair of strings~$\smash{x',y'} \in  \{0,1\}^n$ with $f(x')\ne f(y')$ and $x'=y'^{(i)}$ for some 
  $i\in \diff(x,y)$,
  i.e., a distinguishing pair for the $i$-th variable for some $i\in \diff(x,y)$.
\ignore{
%
}

\begin{figure}[t!]

\fbox{
\begin{minipage}{42em}

\begin{flushleft}
\noindent \textbf{Procedure} \textbf{BinarySearch}$(f,x,y)$\\
\noindent {\bf Input:} Query access to $f \colon \{0, 1\}^n \to \{0, 1\}$
  and two strings $x,y\in \{0,1\}^n$ with $f(x)\ne f(y)$.

\noindent {\bf Output:}  Two strings $x',y'\in \{0,1\}^n$ with $f(x')\ne f(y')$ and 
  $x'=y'^{(i)}$ for some $i\in \diff(x,y)$. 
\begin{enumerate}[itemsep=-0.5mm]
\item Let $B\subseteq [n]$ be the set such that $x=y^{(B)}$.
\item If $|B|=1$ return $x$ and $y$.
\item Partition (arbitrarily) $B$ into $B_1$ and $B_2$ of size $\lfloor |B|/2\rfloor$ and
  $\lceil |B|/2\rceil$, respectively.
\item Query $f(x^{(B_1)})$. 
\item If $f(x)\ne f(x^{(B_1)})$, return \textbf{BinarySearch}$(f,x,x^{(B_1)})$.
\item Otherwise, return \textbf{BinarySearch}$(f,x^{(B_1)},y)$.

\end{enumerate}
\end{flushleft}

\caption{Description of the standard binary search procedure.} \label{fig:binarysearch}

\end{minipage}
}\vspace{0.2cm}
\end{figure}

However, we cannot afford to use the standard binary search procedure
  directly in our main algorithm due to its query complexity of $O(\log n)$;
  recall our goal is to have the query complexity depend on $k$ only.  
Instead we will employ a blockwise version of the 
  binary search procedure,~as described in Figure \ref{fig:blockbinarysearch}.
It takes as input two strings $x,y\in \{0,1\}^n$ with $f(x)\ne f(y)$ and a sequence of pairwise disjoint
  blocks $B_1,\ldots,B_r$ such that $$\diff(x,y)\subseteq B_1\cup \cdots \cup B_r $$ (i.e., $(x,y)$ is a distinguishing pair for $B_1\cup \cdots \cup B_r$), makes $O(\log r)$ queries on $f$, and returns~two strings $x',y'\in \{0,1\}^n$
  satisfying $f(x')\ne f(y')$ and $\diff(x',y')\subseteq B_i$ for some $i\in [r]$
  (i.e., $(x',y')$ is a distinguishing pair for one of the blocks $B_i$ in the input).

\begin{figure}[t!]

\fbox{
\begin{minipage}{42em}

\begin{flushleft}
\noindent \textbf{Procedure} \textbf{BlockBinarySearch}$(f,x,y;B_1,\ldots,B_r)$\\
\noindent {\bf Input:} Query access to $f \colon \{0, 1\}^n \to \{0, 1\}$,
  two strings $x,y\in \{0,1\}^n$ with $f(x)\ne f(y)$,
  and a sequence of pairwise disjoint blocks $B_1,\ldots,B_r$ for some $r\ge 1$ with $\diff(x,y) \subseteq B_1 \cup \cdots \cup B_r$.

\noindent {\bf Output:}  Two strings $x',y'\in \{0,1\}^n$ with $f(x')\ne f(y')$ and 
  $\diff(x,y)\subseteq B_i$ for some $i\in [r]$. 
\begin{enumerate}[itemsep=-0.5mm]
\item If $r=1$ return $x$ and $y$.
\item Let $t=\lfloor r/2\rfloor$ and 
  $B$ be the intersection of $\diff(x,y)$ and $B_1\cup \cdots \cup B_t$.
\item Query $f(x^{(B)})$. 
\item If $f(x)\ne f(x^{(B)})$, return \textbf{BlockBinarySearch}$(f,x,x^{(B)};B_1,\ldots,B_t)$.
\item Otherwise, return \textbf{BlockBinarySearch}$(f,x^{(B)},y; B_{t+1},\ldots,B_r)$.

\end{enumerate}
\end{flushleft}
\caption{Description of the blockwise version of the binary search procedure.} \label{fig:blockbinarysearch}

\end{minipage}
}
\end{figure}



\def\bR{\mathbf{R}}
\def\bP{\mathbf{P}}
\def\bQ{\mathbf{Q}}
\def\bw{\mathbf{w}}
\def\bX{\mathbf{Z}}
\def\bS{\mathbf{S}}
\def\bT{\mathbf{T}}
\def\bC{\mathbf{C}}

\section{Warmup: A tester with $O((k/\eps) + k \log n)$ queries} \label{sec:warmup}

As a warmup, we present in this section a simple, one-sided distribution-free algorithm
  for  testing  $k$-juntas (\textbf{SimpleDJunta}, where the capital letter $D$ is a shorthand for
  distribution-free).~It~uses 
  $O((k/ \eps) + k \log n)$ queries, where $n$ as usual denotes the number of variables
  of the function~being tested.
The idea behind \textbf{SimpleDJunta} and its analysis (Lemma \ref{influence} below)
  will be useful in the next section where we present our main algorithm
  to remove the dependency on $n$.

The algorithm \textbf{SimpleDJunta} maintains a set $I\subset [n]$ 
  which is such that  a distinguishing pair has been found for each $i\in I$
  (i.e., $I$ is a set of relevant variables of $f$ discovered so far).
The algorithm sets $I=\emptyset$ at the beginning and
  rejects only when $|I|$ reaches $k+1$,
  which implies immediately  that the algorithm is one-sided. 
\textbf{SimpleDJunta} proceeds round by round. In each round it draws
  a pair of random strings $\bx$ and $\by$ with $\bx_I=\by_I$.
If $f(\bx)\ne f(\by)$, the standard binary search procedure is used on $\bx$ and $\by$
  to find a distinguishing pair 
  for a new variable $i\in \overline{I}$, which is then added to $I$.

The description of the algorithm can be found in Figure \ref{SimpleDJunta}.
The following theorem establishes its correctness. 

%

\begin{figure}[t!]

\fbox{
\begin{minipage}{42em}

\begin{flushleft}
\noindent \textbf{Algorithm} \textbf{SimpleDJunta}$(f,\D,k,\eps)$\\
\noindent {\bf Input:} Oracle access to a Boolean function $f \colon \{0, 1\}^n \to \{0, 1\}$ and 
  a probability distribution\\ $\D$ over $\{0,1\}^n$,
  a positive integer $k$, and a distance parameter $\eps>0$. 

\noindent {\bf Output:}  Either ``accept'' or ``reject.''
\begin{enumerate}[itemsep=-0.5mm]
\item Set $I=\emptyset$.
\item Repeat {{$ 8({k+1})/{\epsilon}$}} times:
\item \ \ \ \ \ \ Sample $\bx \leftarrow \D$ and a subset $\bR$ of $\overline{I}$ uniformly at random.
  Set $\by=\bx^{(\bR)}$. 
\item \ \ \ \ \ \ {If} $f(\bx)\not=f(\by)$, {then} 
  run the standard binary search on $\bx,\by$ to
   find  a distinguishing
\item \ \ \ \ \ \ \ \ \ \ \ \  pair for a new relevant variable $i\in \bR\subseteq \overline{I}$.
Set $I=I\cup \{i\}$.
\item \ \ \ \ \ \ If $|I|>k$, then halt and output ``reject.''
\item Halt and output ``accept.''
\end{enumerate}
\end{flushleft}
\caption{Description of the distribution-free testing algorithm \textbf{SimpleDJunta} for $k$-juntas.} \label{SimpleDJunta}

\end{minipage}
}

\end{figure}


\begin{theorem}\label{klogn}
(i) The algorithm \emph{\textbf{SimpleDJunta}} makes $O( ({k}/{\epsilon})+k\log n)$ queries and always accepts 
  when $f$ is a $k$-junta.
(ii) It rejects with probability at least $2/3$ if $f$ is $\epsilon$-far from $k$-juntas with respect to~$\calD$. 
\end{theorem}

\begin{proof}
For part (i), note that the algorithm only runs binary search (and spends 
  $O(\log n)$ queries) when $f(\bx)\ne f(\by)$
  and this happens at most $k+1$ times (even though the algorithm has 
  $O(k/\eps)$ rounds).  
The rest of part (i) is immediate from the description of the algorithm.
  
For part (ii), it suffices to show that when $|I|\le k$ at the beginning of a round,
  a new relevant variable is discovered in this round with at least a modestly large probability.
For this purpose we use the following simple but crucial lemma
  and note the fact that $\bx$ and $\by$ on line 3 can be equivalently drawn
  by first sampling $\bx\leftarrow \calD$ and $\bw\leftarrow \{0,1\}^n$ and then setting $\by=\bx_I\circ \bw_{\overline{I}}$
  (the way we draw $\bx$ and $\by$ in Figure \ref{SimpleDJunta} via $\bR\leftarrow
  \overline{I}$ makes
  it easier to connect with the main algorithm in the next section).

\begin{lemma}\label{influence}
If $f$ is $\epsilon$-far from $k$-juntas with respect to $\calD$, then for any  $I\subset [n]$ of size at most $k$, we have
\begin{equation}\label{hahaeq}\Pr_{\bx \leftarrow \calD, \bw \leftarrow \{0,1\}^n}\big[
f(\bx)\not=f(\bx_I\circ \bw_{\overline{I}})\big]\ge {\epsilon}/{2}.\end{equation}
\end{lemma}

Before proving Lemma \ref{influence}, we use it to finish the proof of part (ii).
Assuming Lemma~\ref{influence}~and that $f$ is $\eps$-far from $k$-juntas
  with respect to $\D$, for each round in which $|I| \leq k$ the algorithm finds a 
  new relevant variable with probability at least $\eps/2$.
Using a coupling argument, the~probability that
  the algorithm rejects $f$ (i.e., $|I|$ reaches $k+1$ during the $8(k+1)/\eps$ rounds)
  is at least the probability that $$\sum_{i=1}^{8(k+1)/\eps} \bX_i\ge k+1,$$
  where $\bX_i$'s are i.i.d. $\{0,1\}$-variables that are $1$ with probability $\eps/2$.
It follows from the Chernoff bound that the latter probability is at least $2/3$.
This finishes the proof of the theorem.\ignore{\footnote{\textbf{Xi:} I think you cannot just use the Markov here.
  Even though the expectation is large ($3(k+1)$ before or $4(k+1)$ now) it does not imply that the probability of
  the sum $\ge k+1$ is large.}}
\end{proof}

\begin{proof}[Proof of Lemma \ref{influence}]
Let $I$ be a subset of $[n]$ of size at most $k$.
To prove (\ref{hahaeq}) for $I$, 
we use $I$ to define the following Boolean function $h:\{0,1\}^n\rightarrow \{0,1\}$
  over $n$ variables: for each $x\in \{0,1\}^n$ we set
$$h(x):=\underset{b\in\{0,1\}}{\arg \max}\text{}\left\{\Pr_{\bw\leftarrow \{0,1\}^n}\big[f(x_I\circ \bw_{\overline{I}})=b\big]\right\},$$
where we break ties arbitrarily.
Then for any $x\in \{0,1\}^n$, we have 
\begin{equation}\label{hahaeq2}
\Pr_{\bw\leftarrow\{0,1\}^n}\big[f(x_I\circ\bw_{\overline{I}})=h(x)\big]\ge{1}/{2}.
\end{equation}

Furthermore, we have
\begin{align*}
&\hspace{-1.2cm}\Pr_{\bx \leftarrow \calD, \bw \leftarrow \{0,1\}^n}\big[
f(\bx)\not=f(\bx_I\circ\bw_{\overline{I}})\big]\\[0.6ex]
&=\sum\limits_{z\in\{0,1\}^n}\Pr_{\bx\leftarrow\calD}\big[\bx=z\big] \cdot \Pr_{\bw\leftarrow\{0,1\}^n}\big[
f(z)\not=f(z_I\circ\bw_{\overline{I}})\big]\\[0.6ex]
&\ge \sum\limits_{z\in\{0,1\}^n}\Pr_{\bx\leftarrow\calD}\big[\bx=z\big]
  \cdot \Big((1/2)\cdot \mathbf{1}\big[f(z)\not=h(z)\big]\Big)\\[0.6ex]
&= (1/2)\cdot \Pr_{\bx\leftarrow\calD}\big[f(\bx)\not=h(\bx)\big] \ge \eps/2,
\end{align*}
where the first inequality follows from (\ref{hahaeq2}) and 
  the second inequality follows from the assumption that $f$ is $\eps$-far from 
  every $k$-junta with respect to $\D$ and the fact that $h$ is a $k$-junta
  (since it only depends on variables in $I$ and $|I|\le k$).
This finishes the proof of the lemma.
\end{proof}



\section{Proof of Theorem~\ref{main}:  A tester with $\red{\tilde{O}(k^2)/\eps}$ queries} 
\label{sec:our-algorithm}

In this section, we present our main $\red{\tilde{O}(k^2)/\eps}$-query 
  algorithm for the distribution-free testing of $k$-juntas.
We start with some intuition behind the algorithm.

\subsection{Intuition}\label{sec:intuition}

Recall that the factor of $\log n$ in the query complexity of \textbf{SimpleDJunta}
  from the previous section is due to the use of the standard binary search procedure.
To avoid it, one could choose to terminate each call to binary search early 
  but this ends up giving us relevant \emph{blocks} of variables instead~of relevant variables.
To highlight the challenge, imagine that the algorithm has found so far 
  $\ell\le k$ many pairwise disjoint relevant blocks $B_j$, $j\in [\ell]$, i.e., it has found a distinguishing pair
   for each block $B_j$. By~definition,
  each $B_j$ must contain at least one relevant variable~$i_j\in B_j$. 
However, we do not know exactly  which variable in $B_j$ is $i_j$, and thus it is not clear how to draw a set $\bR$ from $\overline{I}$ uniformly at random,
  where $I=\{i_j:j\in [\ell]\}$, as on line 3 of \textbf{SimpleDJunta},
  in order~to~apply Lemma \ref{influence} to discover a new relevant block.
It seems that we are facing a dilemma when trying to improve \textbf{SimpleDJunta}
  to remove the $\log n$ factor: unless we pin down a set of relevant variables,
  it is not clear how to draw a random set from their complement, but pinning down a single relevant variable
  using the standard binary search procedure
would already cost $\log n$ queries.

To explain the main idea behind our $\red{\tilde{O}(k^2)/\eps}$-query algorithm, let's assume again
  that $\ell\le k$ many disjoint relevant blocks $B_j$ have been found so far,
  with a distinguishing pair $(x^{[j]},y^{[j]})$ for each $B_j$
  (satisfying that $\diff(x^{[j]},y^{[j]})\subseteq B_j$ and 
  $f(x^{[j]})\ne f(y^{[j]})$ by definition). Let 
$$
w^{[j]}=\left(x^{[j]}\right)_{\overline{B_j}}=\left(y^{[j]}\right)_{\overline{B_j}}\in \{0,1\}^{\overline{B_j}}.
$$
Next let us assume further that the function $g_j:=f\uhr_{w^{[j]}}$, for each $j\in [\ell]$, is a \emph{literal}, i.e.~either
  $g_j(z)=z_{i_j}$ for all $z\in \{0,1\}^{B_j}$ or $g_j(z)=\overline{z_{i_j}}$ for all $z\in \{0,1\}^{B_j}$, for some variable $i_j\in B_j$,
  but the variable $i_j$ is of course unknown to the algorithm.
(While this may seems very implausible, we make this assumption for now and explain
  below why it is not too far from real situations.)

To make progress, we draw a random two-way partition of 
  each $B_j$ into $\bP_j$ and $\bQ_j$, i.e., each $i\in B_j$ is added to $\bP_j$ or $\bQ_j$
  with probability $1/2$ (so they are disjoint and $B_j=\bP_j\cup \bQ_j$). 
We make three simple but crucial observations to increase the number of disjoint relevant blocks by one.
\begin{flushleft}\begin{enumerate}
\item Since $g_j$ is assumed to be a literal on the $i_j$-th variable (and by
  the definition of $g_j$ we have query access to $g_j$), it is easy to   
  tell whether $i_j\in \bP_j$ or $i_j\in \bQ_j$, simply by picking an arbitrary string $x \in \{0,1\}^{B_j}$ and comparing $g_j(x)$ with $g_j(x^{(\bP_j)})$.\ignore{
In particular, if $i_j\in \bP_j$, then
  by picking an arbitrary string $x\in \{0,1\}^{B_j}$ and flipping all its bits in $\bP_j$,
  we have that $g_j(x)\ne g_j(x^{(\bP_j)})$ (which is 
  a distinguishing pair for $\bP_j$);
  on the other hand, if $i_j\in \bQ_j$ then this can never happen.}
Below we assume that the algorithm correctly determines whether $i_j$ is in $\bP_j$ or $\bQ_j$ for all $j \in [\ell]$.  We let $\bS_j$ denote the element of $\{\bP_j,\bQ_j\}$ that contains $i_j$ and let $\bT_j$ denote the other one. We also assume below that the algorithm has obtained a distinguishing pair of $g_j$ for each block $\bS_j$.\ignore{figures out where $i_j$ goes for all $j\in [\ell]$,
  letting $\bS_j$ be the one that contains $i_j$ and letting $\bT_j$ be the other,
  and that the algorithm has also found a distinguishing pair of $g_j$ for each block $\bS_j$.}
\item Next we draw a subset $\bT$ of $\overline{B_1\cup\cdots\cup B_\ell}$ uniformly at random.
Crucially, the way that $\bP_j$ and $\bQ_j$ were drawn, and the above assumption that $\bS_j$ contains $i_j$, implies that
\ignore{
Then the way we draw $\bP_j$ and $\bQ_j$ and the assumption that we always figure out
  where $i_j$ goes imply that}
$$
\bR:=\bT\cup \bT_1\cup\cdots \cup \bT_\ell
$$
is indeed a subset of $\overline{I}$ drawn uniformly at random (recall that $I=\{i_j:j\in [\ell]\}$)
  since other than those in $I$, each variable is included in $\bR$ independently 
  with probability $1/2$.
If we draw a random string $\bx\leftarrow \D$, then Lemma \ref{influence} implies that 
  $f(\bx)\ne f(\by)$, where $\by=\bx^{(\bR)}$, with probability at least $\eps/2$.
\item Finally, assuming that $f(\bx)\ne f(\by)$ (with $\diff(\bx,\by)=\bR$),
  running the blockwise binary search
  on $\bx,\by$ and blocks $\bT,\bT_1,\ldots,\bT_\ell$ will lead to 
  a distinguishing pair for one of these blocks and will only require $O(\log \ell) \le O(\log k)$ queries.
If it is a distinguishing pair for $\bT$, then we can add $\bT$ to the list of 
  relevant blocks $B_1,\ldots,B_\ell$ and they remain pairwise disjoint.
If it is $\bT_j$ for some $j\in [\ell]$, then we can replace $B_j$ in the list
  by $\bS_j$ and $\bT_j$, for each of which we have found a distinguishing pair
  (recall that a distinguishing pair has already been found for each $\bS_j$ in the first step).
In either case we have that the number of pairwise disjoint relevant blocks grows by one.
\end{enumerate}\end{flushleft}
  
Coming back to the assumption we made earlier, 
  although $g_j$ is very unlikely to~be a literal, it must fall into one of 
  the following three cases: (1) close to a literal; (2) close to a (all-$0$ or all-$1$)
  constant function; or (3) far from $1$-juntas.  Here in all cases ``close'' and ``far'' means
  with respect to the \emph{uniform distribution} over $\{0,1\}^{B_j}$.
As we discuss in more detail in the rest of the section, with some more careful probability analysis the above arguments generalize~to the case in which every $g_j$ is only close to (rather than exactly) a literal. 
On the other hand, if one of the blocks $B_j$ is in case (2) or (3), then (using the fact that we have a distinguishing pair for $B_j$) it is easy to split $B_j$ into two blocks and find a distinguishing pair for each of them. 
(For example, for case (3) this can be done by running Blais's uniform distribution junta testing algorithm.)
As a result, we can always make progress by increasing the number of 
  pairwise disjoint relevant blocks by one. Our algorithm basically keep repeating
  these steps until the number of such blocks reaches $k+1$.

\subsection{Description of the main algorithm and the proof of correctness}

Our algorithm $\textbf{MainDJunta}(f,\D,k,\eps)$  
  is described in Figure \ref{fig:maindjunta}.
  It maintains two collections of blocks 
$V=\{B_1,\ldots,B_v\}$ ($V$ for ``verified'') and 
  $U=\{C_1,\ldots,C_u\}$ ($U$ for ``unverified'') 
   for some nonnegative integers $v$ and $u$.
They are set to be $\emptyset$ at initialization and always satisfy the following properties:\vspace{0.06cm}
\begin{flushleft}\begin{itemize}
\item[] \textbf{(A)}. $B_1,\ldots,B_v,C_1,\ldots,C_u\subseteq [n]$ are pairwise disjoint
  (nonempty) blocks of variables; \vspace{-0.1cm}
\item[] \textbf{(B)}. A distinguishing pair has been found for each of these blocks.
For notational convenience we use $(x^{[j]},y^{[j]})$ to denote the distinguishing 
  pair for each $B_j$ and $(x^{[C]},y^{[C]})$ to denote the distinguishing pair
  for each block $C\in U$.
We also use the notation 
$$
w^{[j]}:=\left(x^{[j]}\right)_{\overline{B_j}}=\left(y^{[j]}\right)_{\overline{B_j}}\in \{0,1\}^{\overline{B_j}}\quad\text{and}\quad
w^{[C]}:=\left(x^{[C]}\right)_{\overline{C}}=\left(y^{[C]}\right)_{\overline{C}}\in \{0,1\}^{\overline{C}},
$$
and we let $g_j:=f\uhr_{w^{[j]}}$ and $g_C:=f\uhr_{w^{[C]}}$, Boolean functions over 
  $\{0,1\}^{B_j}$ and $\{0,1\}^C$, respectively.
\vspace{0.06cm}
\end{itemize}\end{flushleft}
The algorithm rejects only when the total number of blocks $v+u\ge k+1$ so it is one-sided.

Throughout the algorithm and its analysis, we set a key parameter
$
\gamma := 1/(8k).
$
Blocks in~$V$ are intended to be those that have been ``verified''
  to satisfy the condition that 
  $g_j$ is $\gamma$-close to a literal
  (for some unknown variable $i_j\in B_j$) under the uniform distribution,
  while blocks in $U$ have not been verified yet so they may or may not satisfy the condition.
More formally, at any point in the execution of the algorithm we say that the algorithm is in \emph{good condition} if
  its current collections $V$ and $U$ satisfy conditions \textbf{(A)}, \textbf{(B)} and 
\begin{enumerate}
\item[]\textbf{(C)}. Every $g_j$, $j\in [v]$, is $\gamma$-close to a literal under the uniform 
  distribution over $\{0,1\}^{B_j}$.
\end{enumerate}

The algorithm \textbf{MainDJunta}$(f,k,\eps)$ starts with $V=U=\emptyset$ and 
  proceeds round by round.~For each round, we consider two different types that the round may have:  type 1 is that $u=0$, and type 2 is that $u>0$.
In a type-1 round (with $u=0$) we follow the idea sketched in Section \ref{sec:intuition}
  to increase the total number of disjoint relevant blocks by one.
We prove the following lemma for this case in Section \ref{sec:prooflemma1}.

\begin{lemma}\label{mainlemma1}
Assume that \emph{\textbf{MainDJunta}} is in good condition at the beginning of a round,
  with $u=0$ and $v\le k$.
Then it must remain in good condition at the end of this round.
Moreover, letting $V'$ and $U'$ be the two collections of blocks at the end of this round,
  we have either $|V'|=v$ and $|U'|=1$, or $|V'|=v-1$ and $|U'|=2$
  with probability at least $\eps/4$.
\end{lemma}

\begin{figure}[p!]

\fbox{
\begin{minipage}{42em}

\begin{flushleft}
\noindent \textbf{Algorithm} \textbf{MainDJunta}$(f,\D,k,\eps)$ with 
  the same input\hspace{0.06cm}/\hspace{0.06cm}output as \textbf{SimpleDJunta} in Figure \ref{SimpleDJunta}.\\

\begin{enumerate}[itemsep=0mm]
\item Initialization: Set $V=U=\emptyset$,
  $r_1=\red{64k/\eps}$ and $r_2=\red{3(k+1)}$.\vspace{-0.1cm}
\item While $r_1>0$ and $r_2>0$ do (letting $V=\{B_1,\ldots,B_v\}$ and $U=\{C_1,\ldots,C_u\}$)\vspace{-0.1cm}
\item \ \ \ \ \ If $u=0,$ then \vspace{-0.1cm}
\item \ \ \ \ \ \ \ \ \ \ Set $r_1$ to be $r_1-1$.\vspace{-0.1cm}
\item \ \ \ \ \ \ \ \ \ \ For {$j=1$ to $v$} do
  ($\smash{(x^{[j]},y^{[j]})}$: distinguishing pair
   for $B_j$,
  $\smash{w^{[j]}= (x^{[j]} )_{\overline{B_j}}}$, $\smash{g_j=f\uhr_{w^{[j]}}}$)\vspace{-0.1cm}
\item \ \ \ \ \ \ \ \ \ \ \ \ \ \ \ Draw a random partition $\bP_j,\bQ_j$ of $B_j$ and run \textbf{WhereIsTheLiteral}$(g_j, \bP_j , \bQ_j )$.\vspace{-0.1cm}
\item \ \ \ \ \ \ \ \ \ \ \ \ \ \ \ If it returns a distinguishing pair of $g_j$ for $\bP_j$, set $\bS_j=\bP_j$ and $\bT_j=\bQ_j$;\vspace{-0.1cm}
\item \ \ \ \ \ \ \ \ \ \ \ \ \ \ \ Else if it returns a distinguishing pair of $g_j$ for $\bQ_j$, set $\bS_j=\bQ_j$ and $\bT_j=\bP_j$;\vspace{-0.1cm}
\item \ \ \ \ \ \ \ \ \ \ \ \ \ \ \ Else (it returns ``fail''), skip this round and go back to line 2.\vspace{-0.1cm}
\item \ \ \ \ \ \ \ \ \ \ Draw $\bx\leftarrow \D$ and a subset $\bT$ 
  of $\smash{\overline{B_1\cup\cdots \cup B_v}}$ uniformly at random.\vspace{-0.11cm}
\item 

  \ \ \ \ \ \ \ \ \ \ If {$f(\bx) \neq f(\by)$}, 
  where $\smash{\by=\bx^{(\bR)}}$ with $\bR=\bT\cup \bT_1\cup\cdots\cup\bT_v$,  then\vspace{-0.1cm}
\item \ \ \ \ \ \ \ \ \ \ \ \ \ \ \ Run the blockwise binary search on $\bx$ and $\by$
  with blocks $\bT,\bT_1,\ldots,\bT_v$.\vspace{-0.1cm}
\item \ \ \ \ \ \ \ \ \ \ \ \ \ \ \ If a distinguishing pair of $f$ for $\bT$ is found, add $\bT$ to $U$.\vspace{-0.1cm}
\item \ \ \ \ \ \ \ \ \ \ \ \ \ \ \ Else (a distinguishing pair of $f$ for $\bT_{j^*}$, for some $j^*\in [v]$,
  is found)\vspace{-0.1cm}
\item \ \ \ \ \ \ \ \ \ \ \ \ \ \ \ \ \ \ \ \ Concatenate $\smash{w^{[j^*]}}$ to the distinguishing 
  pair of $g_{j^*}$
  for $\bS_{j^*}$ found on line 7-8.\vspace{-0.1cm}
\item \ \ \ \ \ \ \ \ \ \ \ \ \ \ \ \ \ \ \ \ This gives us a distinguishing pair of $f$ for $\bS_{j^*}$.\vspace{-0.1cm}
\item \ \ \ \ \ \ \ \ \ \ \ \ \ \ \ \ \ \ \ \ Remove $B_{j^*}$ from $V$ and add both
  $\bS_{j^*}$ and $\bT_{j^*}$ to $U$.\vspace{-0.1cm}


\item \ \ \ \ \ Else (i.e., $u>0$)\vspace{-0.1cm}
\item \ \ \ \ \ \ \ \ \ \ Set $r_2$ to be $r_2-1$.\vspace{-0.1cm}
\item \ \ \ \ \ \ \ \ \ \  Pick a $C\in U$ arbitrarily; let $(x,y)$ be its distinguishing pair,
  $w=x_{\overline{C}}$ and $g=f\uhr_w$. \vspace{-0.1cm}
\item \ \ \ \ \ \ \ \ \ \ If \textbf{Literal}$(g)$ returns ``true,'' remove $C$ from $U$ and add it to $V$.\vspace{-0.1cm}
\item \ \ \ \ \ \ \ \ \ \ Else (it returns disjoint subsets $C',C^*$ of $C$
  and each a distinguishing pair of $g_C$)\vspace{-0.1cm}
\item \ \ \ \ \ \ \ \ \ \ \ \ \ \ \ Concatenate $w$ to obtain a distinguishing pair 
  of $f$ for each of $C'$ and $C^*$\vspace{-0.1cm}
\item \ \ \ \ \ \ \ \ \ \ \ \ \ \ \ Remove $C$ from $U$ and add both $C'$ and $C^*$ to $U$.\vspace{-0.1cm}
\item \ \ \ \ \ If $|V|+|U|\ge k+1,$ then halt and output ``reject.''\vspace{-0.1cm}
    \item Halt and output ``accept.'' \vspace{-0.2cm}
\end{enumerate}\end{flushleft}
\caption{Description of the distribution-free testing algorithm \textbf{MainDJunta} for $k$-juntas.}
\label{fig:maindjunta}

\end{minipage}
}
\end{figure}

\begin{figure}[t!]

\fbox{
\begin{minipage}{42em}

\begin{flushleft}
\noindent \textbf{Subroutine} \textbf{WhereIsTheLiteral}$(g,P,Q)$\\
\noindent {\bf Input:} Oracle access to a Boolean function $g$ over $\{0,1\}^B$
  with $P,Q$ being a partition of $B$.

\noindent {\bf Output:}  Either a distinguishing pair for $P$, 
  a distinguishing pair for $Q$, or ``fail.''
\begin{enumerate}[itemsep=-0.5mm]
\item Draw $\bw\leftarrow \{0,1\}^Q$ and $\bz \leftarrow \{0,1\}^P$ independently and uniformly at random.
\item If $g(\bw\circ\bz )\ne g(\bw\circ\bz^{(P)})$, return $(\bw\circ\bz ,\bw\circ\bz^{(P)})$
  as a distinguishing pair for $P$.
\item Draw $\bw'\leftarrow \{0,1\}^P$ and $\bz' \leftarrow \{0,1\}^Q$ independently and uniformly at random.
\item If $g(\bw'\circ\bz')\ne g(\bw'\circ\bz'^{(Q)})$, return $(\bw'\circ\bz ',\bw'\circ\bz'^{(Q)})$
  as a distinguishing pair for $Q$.
\item Return ``fail.''\vspace{-0.3cm}
\end{enumerate}\end{flushleft}
\caption{Description of the subroutine \textbf{WhereIsTheLiteral}.}\vspace{0.25cm}\label{fig:where}

\end{minipage}
}\vspace{0.2cm}
\end{figure}




\begin{figure}[h!]

\fbox{
\begin{minipage}{42em}

\begin{flushleft}
\noindent \textbf{Subroutine} \textbf{Literal}$(g)$\\
\noindent {\bf Input:} Oracle access to a Boolean function $g$ over $\{0,1\}^C$ where $C$ has a distinguishing pair.

\noindent {\bf Output:}  ``True'' or disjoint nonempty subsets 
  $C',C^*$ of $C$ and a distinguishing pair for each.
\begin{enumerate}[itemsep=-0.5mm] 
\item Repeat $\log k+\red{6}$ times:
\item \ \ \ \ \ If {\textbf{UniformJunta}\((g,1,\gamma)\)} rejects, then
\item \ \ \ \ \ \ \ \ \ \ Return the two disjoint blocks it has found and a distinguishing pair
  for each.
\item Let $(x,y)$ be the distinguishing pair for $C$.
\item Repeat $\log k+3$ times:
\item \ \ \ \ \ Draw a random partition $\bC',\bC^*$ of $C$ and query 
  $g(x^{(\bC')}),g(x^{(\bC^*)}),g(y^{(\bC')}),g(y^{(\bC^*)})$.
\item \ \ \ \ \ If $\smash{g(x^{(\bC')})=g(x^{(\bC^*)})\ne g(x)}$, then
\item \ \ \ \ \ \ \ \ \ \ Return $\bC',\bC^*$ and $(x,x^{(\bC')})$ and $(x,x^{(\bC^*)})$ as their distinguishing pairs.
\item \ \ \ \ \ If $\smash{g(y^{(\bC')})=g(y^{(\bC^*)})\ne g(y)}$, then
\item \ \ \ \ \ \ \ \ \ \ Return $\bC',\bC^*$ and $(y,y^{(\bC')})$ and $(y,y^{(\bC^*)})$ as their distinguishing pairs.


\item Return ``true.''\vspace{-0.3cm}
\end{enumerate}\end{flushleft}

\caption{Description of the subroutine \textbf{Literal}.}\label{fig:literal}
\end{minipage}
}
\end{figure}

In a type-2 round (with $u\ge 1$), we pick an arbitrary block $C$ from $U$ and  
  check whether $g_C$ is close to a literal under~the uniform distribution.
The following lemma, which we prove in Section \ref{sec:prooflemma2},
  shows that with high probability, either $C$ is moved to collection $V$ and 
  the algorithm remains in good condition, or the algorithm finds two disjoint subsets 
  of $C$ and a distinguishing pair for each of them so that $V$ stays the same but $|U|$ goes up by one
  (we add these two blocks to $U$ since they have not yet been verified).

\begin{lemma}\label{mainlemma2}
Assume that \emph{\textbf{MainDJunta}} is in good condition at the beginning of a round,
  with $u>0$ and $v+u\le k$.
Then with probability at least $1-1/\red{(64k)}$, one of the following two events occurs at the end of this round
  (letting $V'$ and $U'$ be the two collections of blocks at the end of this round):
\begin{enumerate}
\item The algorithm remains in good condition with $|V'|=|V|+1$ and $|U'|=|U|-1$; or\vspace{-0.15cm}
\item The algorithm remains in good condition with $V'=V$ and $|U'|=|U|+1$.
\end{enumerate}
\end{lemma}

Assuming Lemma \ref{mainlemma1} and Lemma \ref{mainlemma2}, we are 
  ready to prove the correctness of $\textbf{MainDJunta}$.
  
\begin{theorem}\label{maincorrectness}
(i) The algorithm \emph{\textbf{MainDJunta}} makes $\red{\tilde{O}(k^2)/\eps}$ queries 
  and always accepts $f$ when it is a $k$-junta.
(ii) It rejects with probability at least $2/3$ when $f$ is $\eps$-far from every $k$-junta with respect to $\D$.
\end{theorem}
\begin{proof}[Proof of Theorem \ref{maincorrectness} Assuming Lemmas \ref{mainlemma1} and \ref{mainlemma2}]
 \textbf{MainDJunta} is one-sided since it rejects $f$ only when it has found 
  $k+1$ pairwise disjoint relevant blocks of $f$.
Its query complexity is 
\begin{align*}
&\hspace{-0.5cm}\text{(\# type-1 rounds)\hspace{0.05cm}$\cdot$\hspace{0.05cm}(\# queries per type-1 round) $ + $ (\# type-2 rounds)\hspace{0.05cm}$\cdot$\hspace{0.05cm}(\# queries per type-2 round)}\\[0.3ex]
&= O(k/\eps)\cdot (O(k)+O(\log k))+ O(k)\cdot O(\log k)\cdot O(k)
=O(k^2/\eps)+O(k^2\log k)=O(k^2\log k)/\eps.
\end{align*}
In the rest of the proof 
we show that it rejects $f$ with probability at least $2/3$ when $f$ is $\eps$-far from every $k$-junta
  with respect to $\D$.

For this purpose we introduce a simple potential function $F$ to measure the progress:
$$
F(V,U):=\red{3|V|+2|U|}.
$$
Each round of the algorithm is either of type-$1$ (when $|U|=0$) or of type-$2$ (when $|U|>0$).~By Lemma \ref{mainlemma1}, if the algorithm is in good condition at
  the beginning of a type-1 round, then
  the algorithm ends the round in good condition and the potential function $F$~goes~up by \red{at least~one} with probability at least $\eps/4$ (in which case we say that 
  the algorithm succeeds in this \mbox{type-$1$} round).  By Lemma  \ref{mainlemma2}, if the algorithm is in good condition at the
  beginning of a type-$2$ round, then the algorithm ends the round in good condition and 
   $F$ goes up by \red{at least one} with probability at least $1-1/(32k)$
  (in which case we say it succeeds in this type-$2$ round).
  
Note that $F$ is $0$ at the beginning ($V=U=\emptyset$)  and that we must have $|U|+|V|\ge k+1$
  (and thus, the algorithm rejects) when the potential function $F$ reaches $\red{3(k+1)}$ or above.
As a result,  a necessary condition for the algorithm to accept is that one of the following 
  two events occurs:
\begin{flushleft}\begin{enumerate}
\item[] $E_1$: At least one of the (no more than $\red{3(k+1)}$ many)  type-$2$ rounds fails.\vspace{-0.12cm}  
\item[] $E_2$: $E_1$ does not occur (so the algorithm ends every round in good condition, and the reason
  that the algorithm accepts cannot be that it uses up all the $3(k+1)$ many type-$2$ rounds), and
  the algorithm uses up all the $\red{64}k/\eps$ many type-$1$ rounds but at most 
  $3k+2$\\ of them succeed.
\end{enumerate}\end{flushleft}
By a union bound, the probability of $E_1$ is at most 
$$\red{3(k+1)\cdot 1/(64k)\le 6k\cdot 1/(64k)<1/8.}$$
As the algorithm ends every round in good condition, it follows from Lemma~\ref{mainlemma1} from a coupling argument
  that the probability of $E_2$ is at most the probability that
$$\sum_{i=1}^{\red{64k/\eps}} \bZ_i\le \red{3k+2},$$ where $\bZ_i$'s are i.i.d. $\{0,1\}$-valued random variables
  that take $1$ with probability $\eps/4$. It follows from the Chernoff bound the probability is at most
    (using $3k+2\le 5k$)
$$
\red{\exp\left(-\left(\frac{11}{16}\right)^2\cdot \frac{16k}{2}\right)=\exp\left(-\frac{121k}{32}\right)
< \exp(-3)<1/8.}
$$
Finally it follows from a union bound that the algorithm rejects with probability at least $3/4$.
\end{proof}

\subsection{Proof of Lemma \ref{mainlemma1}}\label{sec:prooflemma1}

We start with a lemma for the subroutine \textbf{WhereIsTheLiteral}, which
  is described in Figure \ref{fig:where}.

\begin{lemma}\label{lem:where}
Assume that $g:\{0,1\}^B\rightarrow \{0,1\}$ is $\gamma$-close (with respect to the uniform distribution) to a literal $x_i$ or $\overline{x_i}$ for some $i\in B$.~If $i\in P$, then \emph{\textbf{WhereIsTheLiteral}}$(g,P,Q)$ returns a distinguishing pair of $g$ for $P$
  with probability at least $1-4\gamma$; If $i\in Q$, then it returns a distinguishing pair of $g$ for $Q$
  with probability at least $1-4\gamma$.
\end{lemma}
\begin{proof}
Let $K$ be the set of strings $x\in \{0,1\}^B$ such that $g(x)$ disagrees with the literal to
  which~it~is $\gamma$-close (so $|K|\le \gamma \cdot 2^{|B|}$).
We work on the case when $i\in Q$; the case when $i\in P$ is similar.

By the description of \textbf{WhereIsTheLiteral}, it returns a distinguishing pair for $Q$
  if
$$g(\bw\circ\bz )=g(\bw\circ\bz^{(P)})\quad\text{and}\quad
g(\bw'\circ\bz')\ne g(\bw'\circ\bz'^{(Q)}).$$
Note that this holds if all four strings fall outside of $K$ and thus,
  the probability that it does not hold is at most the probability
  that at least one of these four strings falls inside $K$.
The latter by a union bound is at most $4\gamma$ since each of these four strings 
  is drawn uniformly at random from $\{0,1\}^B$ by itself.
This finishes the proof of the lemma.
\end{proof}

We are now ready to prove Lemma \ref{mainlemma1}.

\begin{proof}[Proof of Lemma \ref{mainlemma1}]
First, it is easy to verify that if the algorithm starts a round in good condition, then it ends it in good condition.
This is because whenever a block is added to $U$, it is disjoint from other blocks and we have
  found a distinguishing pair for it.
  
Next it follows directly from Lemma \ref{lem:where} and a union bound that,
  for any sequence of partitions $P_j$ and $Q_j$ of $B_j$ picked on line 6,
  the probability that the for-loop correctly sets $\bS_j$ to be 
  the one that contains the special variable $i_j$ for all $j\in [v]$ is at least
  (recalling that $\gamma=1/(8k)$)
$$
1-4\gamma \cdot v \ge 1-4\gamma\cdot k= 1/2.
$$
Now we can view the process equivalently as follows.
First we draw $\bx\leftarrow \D$, $\bT\leftarrow \overline{B_1\cup\cdots \cup B_v}$,
  and random partitions $\bP_j,\bQ_j$ of each $B_j$.
If we let $\bT^*_j$ be the set in $\bP_j,\bQ_j$ that does not contain the special variable,
  then $\bR^*=\bT\cup\bT^*_1\cup\cdots\cup\bT^*_v$ is a set drawn uniformly at random
  from $\overline{I}$, where $I=\{i_j:j\in [v]\}$ consists of the special variables.
Therefore, it follows from Lemma \ref{influence}
  that $f(\bx)\ne f(\bx^{(\bR^*)})$ with probability at least $\eps/2$.
Since with probability at least $1/2$, the set $\bR$ 
  on line 11 agrees with $\bR^*$,
  we have that the algorithm reaches line 12 with
  $f(\bx)\ne f(\by)$ with probability at least $\eps/4$.
  Given this, the lemma is immediate by inspection of lines 12-17 of the algorithm.
\end{proof}

\subsection{Proof of Lemma \ref{mainlemma2}}\label{sec:prooflemma2}

First it follows from the description of the subroutine \textbf{Literal}$(g)$ that
  it either returns ``true'' or~a pair of nonempty disjoint subsets $C',C^*$ of $C$
  and a distinguishing pair of $g$ for each of them (see Theorem \ref{blaistheorem}).
Next, let $C \in V$ be the block picked in line 20.
If $g$ is $\gamma$-close to a literal, then it is easy to verify that 
  one of the two events described
  in Lemma \ref{mainlemma2} must hold (using the property of \textbf{Literal}$(g)$ above).
So we focus on the other two cases in the rest of the proof:
  $g$ is $\gamma$-far from $1$-juntas or $g$ is $\gamma$-close to a (all-$1$ or all-$0$)
  constant function.
In both cases we show below that the second event happens with high probability.

When $g$ is $\gamma$-far from $1$-juntas under the uniform distribution, we have that
  one of the $\red{\log k+6}$ calls to \textbf{UniformJunta} in \textbf{Literal} rejects with probability at least
$$
1-(1/3)^{ \log k+\red{6}}> 1-1/(\red{64}k).
$$ 
The second event in Lemma \ref{mainlemma2} occurs when this happens.

When $g$ is $\gamma$-close to a constant function (say the all-$1$ function),
  we have that either string $x$ or $y$ in the distinguishing pair for $C$
  disagrees with the function (say $g(x)=0$, since $g(x)\ne g(y)$).
Let $K$ be the set of strings in $\{0,1\}^C$ that disagree with the all-$1$ function.
Then line 7 of \textbf{Literal}$(g)$ does not hold only when one of $x^{(\bC')}$ or
  $x^{(\bC^*)}$ lies in $K$.
As both strings are distributed uniformly over $\{0,1\}^C$ by themselves, this happens 
  with probability at most $2\gamma$ by a union bound.
Therefore the probability that line 7 holds at least once is at least
$$
1-(2\gamma)^{\log k+3}=1-(1/(4k))^{\log k+3}>
1-(1/4)^{\log k+3}=
 1-1/(64k^2).
$$ 
As a result, the second event in Lemma \ref{mainlemma2} occurs with probability at least $1-1/(64k^2)$.

This finishes the proof of Lemma \ref{mainlemma2}.

\def\bD{\boldsymbol{\mathcal{D}}}
\def\E{\mathbf{E}} \def\balpha{\boldsymbol{\alpha}} \def\bbeta{\boldsymbol{\beta}}
\def\calE{\mathcal{E}} \def\bgamma{\boldsymbol{\gamma}}
\def\larr{\leftarrow} \def\calJ{\mathcal{J}}

\section{Proof of Theorem~\ref{lowerbound}:  An $\Omega(2^{k/3})$-query non-adaptive 
  lower bound}\label{sec:lower}

In this section we prove the $\Omega(2^{k/3})$ lower bound for the non-adaptive
  distribution-free testing of $k$-juntas that was stated as Theorem~\ref{lowerbound}. 
We start with some notation. 
Given a sequence $Y=(y^i:$ $i\in [q])$ of $q$ strings in $\{0,1\}^n$ and 
  a Boolean function $\phi:\{0,1\}^n\rightarrow \{0,1\}$, we write $\phi(Y)$ to denote the 
  $q$-bit string $\alpha$ with $\alpha_i=\phi(y^i)$ for each $i\in [q]$.
We also write $\bY=(\by^i:i\in [q])\larr \D^q$ to denote 
  a sequence of $q$ independent draws from the same probability distribution $\D$.

Let $k$ and $n$ be two positive integers that satisfy $k\le n/200$.
We may further assume that $k$~is at least some absolute constant $C$ (to be specified 
  later) since otherwise, the claimed $\Omega(2^{k/3})$ lower bound on query complexity holds trivially
  due to the constant hidden behind the $\Omega$. 
Let~$q=2^{k/3}.$
For convenience we refer to an algorithm as a $q$-query algorithm if it makes $q$ sample queries
  and $q$ black-box queries \emph{each}. Such algorithms are clearly at least as powerful as
  those that make $q$ queries in total. 
Our goal is then to show that there exists no $q$-query non-adaptive (randomized) algorithm 
  for the distribution-free testing of $k$-juntas over Boolean functions of $n$ variables, even when the distance parameter $\eps$ is $1/3$.
  
By Yao's minimax principle we focus on $q$-query non-adaptive \emph{deterministic} algorithms.    
Such an algorithm $A$ (which consists of two deterministic maps $A_1$ and $A_2$
  as discussed below) works as follows.  
Upon an input pair $(\phi,\D)$, 
  where $\phi:\{0,1\}^n\rightarrow \{0,1\}$ and $\D$ is a probability distribution over $\{0,1\}^n$,
  the algorithm receives in the first phase a sequence $Y=(y^{i}:i\in [q])$ of $q$
  strings (which should be thought of as samples from $\D$) and a binary string $\alpha=\phi(Y)$ of length $q$.
In the second phase, the algorithm $A$ uses the first map $A_1$ to obtain   
  a sequence of  $q$ strings $Z=(z^{i}:i\in [q])=A_1(Y,\alpha)$, and 
  feeds them to the black-box oracle.
Once the query results $\beta=\phi(Z)$ are back, $A_2(Y,\alpha,\beta)$ returns either $0$ or $1$ (notice that we do not need to include $Z$
  as an input of $A_2$ since it is determined by $Y$ and $\alpha$) in which cases the algorithm
  $A$ either rejects or accepts, respectively. 
A randomized algorithm $T$ works similarly and consists of two similar maps $T_1$ and $T_2$ but both
  are randomized. 

Given the description above, unlike typical deterministic algorithms,
  whether $A$ accepts or~not depends on not only $(\phi,\D)$ 
  but also the sample strings $\bY\leftarrow \D^q$ it draws.
Formally we have
$$
\Pr\big[\text{$A$ accepts $(\phi,\D)$}\big]=\Pr_{\bY\leftarrow \D^q}\Big[A_2\Big(\bY,\phi(\bY),\phi\big(A_1(\bY,\phi(\bY))\big) \Big)=1\Big].
$$

The plan of the rest of the section is as follows. We define in Section~\ref{sec:distributions} a pair of probability distributions $\YES$ and $\NO$ over pairs $(\phi,\D)$, where $\phi$ is a Boolean function over $n$ variables~and $\D$ is a distribution over $\{0,1\}^n$.
\red{For clarity we use $(f,\D)$ to denote pairs in the support of~$\YES$ and 
  $(g,\D)$ to denote pairs in the support of $\NO$.}
We show that (1) Every  $(f,\D)$ in the support of $\YES$ satisfies~that $f$ is a $k$-junta
  (Lemma \ref{lem:yes:simple});
(2) With probability $1-o_k(1)$, $(\bg,\bD)\leftarrow \NO$ satisfies that $\bg$ is~$1/3$-far from every $k$-junta with respect to $\bD$ (Lemma \ref{lem:far}). 
\ignore{
We first give definitions for $\YES$ distribution and $\NO$ distribution in sections \ref{sec:yes}, \ref{sec:no}, such that every draw $(f,\calD_f)$ from $\YES$ satisfies that $f$ is a $k$-junta, while with high probability a draw $(g,\calD_g)$ from $\NO$ satisfies that $g$ is $1/2$-far from $k$-junta with respect to $\calD_g$.
}
To obtain Theorem~\ref{lowerbound},~it suffices to prove the following 
  main technical lemma, which informally says that any $q$-query non-adaptive deterministic
    algorithm must behave similarly when it is run on $(\boldf,\bD) \leftarrow \YES$ versus $(\bg,\bD) \leftarrow \NO$:

\ignore{
}

\begin{lemma}\label{lem:main-lower}
Any $q$-query deterministic algorithm
  $A$ satisfies
\begin{equation}\label{hehe1}
\mathlarger{\Big|} \hspace{0.08cm}\E_{(\boldf,\bD ) \leftarrow \YES}\Big[ \Pr\big[\text{$A$ accepts $(\boldf,\bD)$}\big]\Big]
- \E_{(\bg,\bD)\leftarrow \NO} \Big[\Pr\big[\text{$A$ accepts $(\bg,\bD)$}\big]\Big] 
\hspace{0.06cm}\mathlarger{\Big|}\le {1/4}. \end{equation}
\end{lemma}
\begin{proof}[Proof of Theorem \ref{lowerbound} Assuming Lemma \ref{lem:main-lower}, 
\ref{lem:yes:simple} and \ref{lem:far}]
Assume for a contradiction that there exists a $q$-query non-adaptive randomized algorithm $T$
  for the distribution-free testing of $k$-juntas over $n$-variable Boolean functions when $\eps=1/3$.
Then it follows from Lemma \ref{lem:yes:simple} and \ref{lem:far} that
$$
\E_{(\boldf,\bD)\leftarrow \YES}\Big[\Pr\big[\text{$T$ accepts $(\boldf,\bD)$}\big]\Big]
- \E_{(\bg,\bD)\leftarrow\NO}\Big[\Pr\big[\text{$T$ accepts $(\bg,\bD)$}\big]\Big]
 \ge 1/3-o_k(1),$$
since the first expectation is at least $2/3$ and the second is at most $$1/3\big(1-o_k(1)\big)+ o_k(1)
  \le 1/3 +o_k(1).$$
As $T$ is a probability distribution over deterministic algorithms,
  there must exist~a~\mbox{$q$-query} non-adaptive deterministic algorithm $A$ that satisfies
$$
\E_{(\boldf,\bD)\leftarrow \YES}\Big[\Pr\big[\text{$A$ accepts $(\boldf,\bD)$}\big]\Big]
- \E_{(\bg,\bD)\leftarrow\NO}\Big[\Pr\big[\text{$A$ accepts $(\bg,\bD)$}\big]\Big]
 \ge 1/3-o_k(1),$$
  a contradiction with Lemma \ref{lem:main-lower} when $k$ is sufficiently large.
\end{proof}

\subsection{The $\YES$ and $\NO$ distributions} \label{sec:distributions}


Given $J\subseteq [n]$, we partition $\{0,1\}^n$ into \emph{sections} (with respect to $J$)  
  where the $z$-section, $z\in \{0,1\}^J$, consists of those $x\in \{0,1\}^n$ which have $x_J=z$.
We  write
$\rjunta_J$ to denote the uniform distribution over all juntas over $J$.  More precisely,
a Boolean function $\bh:\{0,1\}^n\rightarrow \{0,1\}$ drawn from~$\rjunta_J$ is
  generated as follows:
For each $z\in \{0,1\}^J$, a bit $\bb(z)$ is chosen independently and uniformly at random, and for each $x \in \{0,1\}^n$ the value of $\bh(x)$ is set to $\bb(x_J)$.
Let 
$$m := 36 \cdot 2^k \ln n.$$

We start with $\YES$.
A pair $(\boldf,\bD)$ drawn from $\YES$ is generated as follows:
\begin{flushleft}\begin{enumerate}
\item First we draw independently a subset $\bJ$ of $[n]$ of size $k$ uniformly at random
  and a subset $\bS$ of $\{0,1\}^n$ of size $m$ uniformly at random.

\item Next we draw $\boldf\larr\rjunta_\bJ$ and set 
$\bD$ to be the uniform distribution over $\bS$.
\end{enumerate}\end{flushleft}
For technical reasons that will become clear in Section \ref{sec:main-lower} 
  we use $\YES^*$ to denote the probability distribution supported over triples $(f,\D,J)$, with 
  $(\boldf,\bD,\bJ)\larr \YES^*$ being generated by the same two steps above (so the only
  difference is that we include $\bJ$ in elements of $\YES^*$).
  
The following observation is straight-forward from the definition of $\YES$.

\begin{lemma}\label{lem:yes:simple}
The function $f$ is a $k$-junta for every pair $(f,\D)$ in the support of $\YES$.
\end{lemma}

\ignore{  first draw a set of $k$ indices $J$ uniformly
  at random from $[n]$.
We also draw a set of $m$ binary strings $S$ from
  $\{0,1\}^n$ uniformly at random, where
  $m = 2^{k+2}$.
Let $\calD_f$ be the uniform distribution over $S$ and let $f$ be a function drawn from the distribution $\rjunta_J$.
Then $f$ is clearly a $k$-junta following the definition.
}


We now describe $\NO.$  A pair $(\bg,\bD)$ drawn from $\NO$ is generated as follows:
\begin{flushleft}\begin{enumerate}
\item 
We draw $\bJ$ and $\bS$ in the same way as the first step of $\YES$.
\item
Next we draw $\bh\larr \rjunta_\bJ$ and a map $\bgamma:\bS\rightarrow\{0,1\}$ uniformly
  at random by choosing a bit independently and uniformly at random for each string in $\bS$.
We usually refer to $\bh$ as the ``\emph{background junta}.''
\item The distribution $\bD$ is set to be the uniform distribution over $\bS$, which is the 
  same as $\YES$.
The function $\bg: \{0,1\}^n \to \{0,1\}$
is defined using $\bh,\bS$ and $\bgamma$ as follows:
\begin{itemize}
\item [(a)] For each string $y \in \bS$, set $\bg(y)=\bgamma(y)$;\vspace{0.06cm}
\item [(b)] For each string $x\notin \bS$, if there exists no $y\in \bS$ with $\red{y_\bJ=x_\bJ}$ and $d(x,y)\le 0.4n$, set
  $\bg(x)=\bh(x)$; otherwise we set $\bg(x)=1$ if there exists such a $y\in \bS$ with $\bgamma(y)=1$,
  and set $\bg(x)=0$ if every such $y\in \bS$ has $\bgamma(y)=0$.
(The choice of the tie-breaking rule here is not important; we just pick one to  
  make sure that $\bg$ is well defined in all cases.)
\end{itemize}
\end{enumerate}\end{flushleft}
Similarly we let $\NO^*$ denote the distribution supported on triples $(g,\D,J)$ as generated above.

See Figure~\ref{fig:example-x} for an illustration of a function drawn from $\NO$.
To gain some intuition, we first note that about half of the 
  strings $z\in S$ have $g(z)$ disagree with the value of the background junta on the section it lies in.
With such a string $z$ in hand (from one of the samples received in the first phase),
  an algorithm may attempt to find a string $w$ that lies in the same section as $z$ but
  satisfies $g(z)\ne g(w)$.
If such a string is found, the algorithm knows for sure that $(g,\D)$ is from the $\NO$ distribution. 
However, finding such a  $w$ is not easy because one must flip more than
  $0.4n$ bits of $z$, but without knowing the variables in $J$ it is hard to keep $w$ in the same section as $z$ after flipping this many bits.

Next we prove that with high probability,
$(\bg,\bD)\larr \NO$ satisfies that 
  $\bg$ is $1/3$-far from every $k$-junta with respect to $\bD$:

\begin{figure}
\centering
\includegraphics[width=0.45\linewidth]{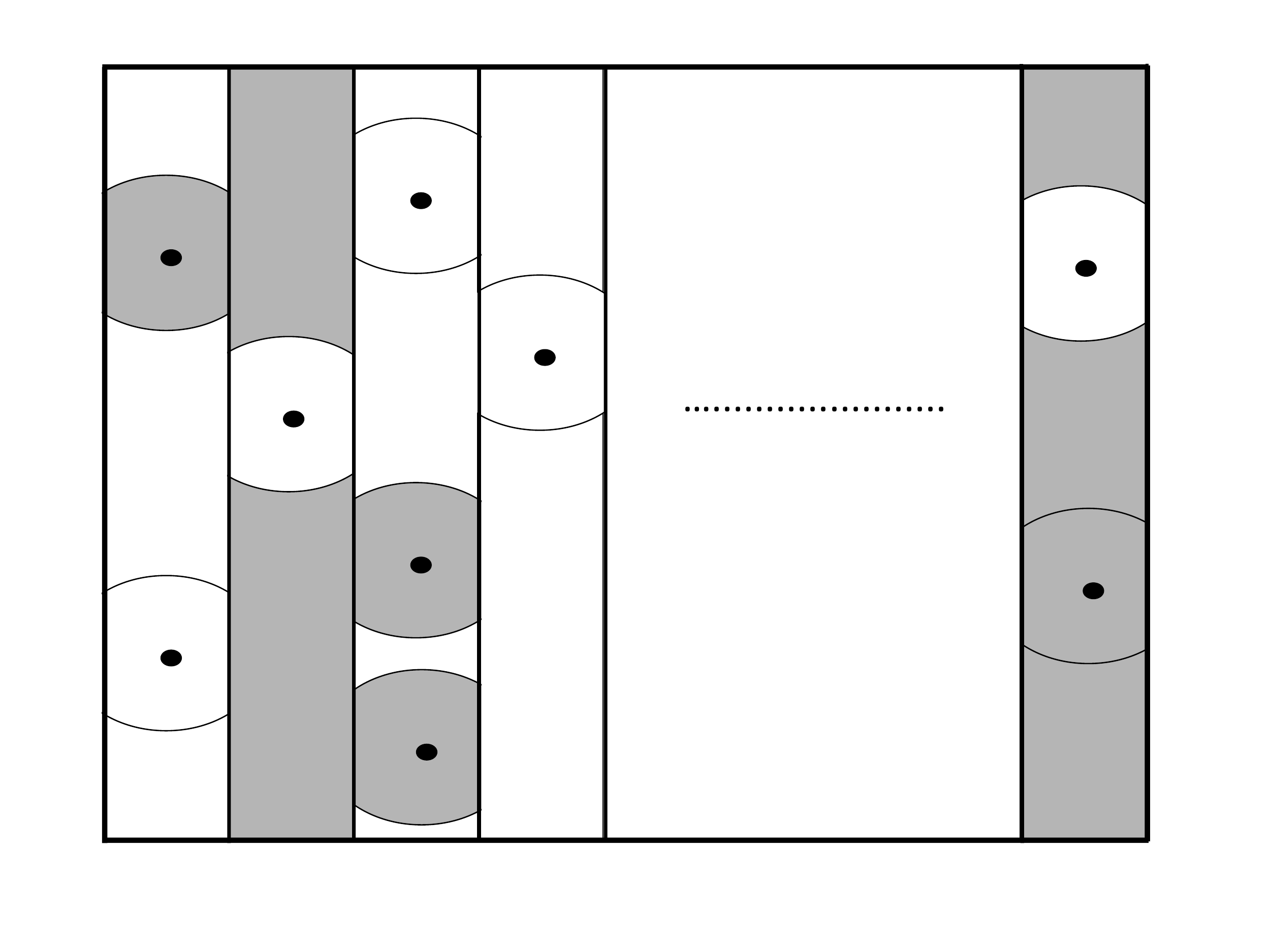}
\caption{A schematic depiction of how $\{0,1\}^n$ is labeled by a function $g$ from $\NO$.  The domain $\{0,1\}^n$ is partitioned into $2^k$ \emph{sections} corresponding to different settings of the variables in $J$; each section is a vertical strip in the figure.  Shaded regions correspond to strings where $g$
  evaluates~to $1$ and unshaded regions to strings where $g$ evaluates to $0$.  Each string in $S$ is a black dot and the value of $g$ on each such string is chosen uniformly at random.  Since in this figure the truncated circles are disjoint, the tie-breaking rule does not come into effect, and for each $z \in S$ all strings in its section within distance at most $0.4n$ (the points in the truncated circle around $z$) have the same value as $z$. The value of $g$ on other points is determined by the background junta $h$ which assigns a uniform random bit to each section.
}
\label{fig:example-x}
\end{figure}


\ignore{

}

\begin{lemma} \label{lem:far}
With probability at least $1-o_k(1)$, $(\bg,\bD)\larr \NO$ is such that 
  $\bg$ is $1/3$-far from every $k$-junta with respect to the distribution $\bD$ .
\end{lemma}

\begin{proof}
\ignore{
}
Fix a $k$-junta $h$, i.e. any set $I \subset [n]$ with $|I|=k$ and any $2^k$-bit truth table over  variables~in~$I$.
We have that $\dist_{\bD}(\bg,h)$ is precisely the fraction of strings $z \in \bS$ such that 
  $\bgamma(z) \neq h(z).$ Since each bit $\bgamma(z)$ is drawn independently 
  and uniformly at random, we have that
\[
\Pr_{(\bg,\bD ) \leftarrow \NO}\big[\hspace{0.06cm}\dist_{\bD }(\bg,h) \leq 1/3
\hspace{0.06cm}\big] = \Pr_{\bj \leftarrow \Bin(m,1/2)}\big[\hspace{0.06cm}\bj \leq m/3
\hspace{0.06cm}\big],
\]
which, recalling that $m=36 \cdot 2^k \ln n$, by a standard Chernoff bound is at most $e^{-m/36} = n^{-2^k}.$  The result follows by a union bound over all (at most) $${n \choose k}\cdot 2^{2^k} \leq n^k\cdot  2^{2^k}=o_k(n^{2^k})$$ possible $k$-juntas $h$ over $n$ variables.
This finishes the proof of the lemma.
\end{proof}

Given Lemma~\ref{lem:yes:simple} and \ref{lem:far}, to prove Theorem~\ref{lowerbound} it remains only to prove Lemma~\ref{lem:main-lower}.

\ignore{

I think in the old argument here, if we fix a junta $h$ and then argue about $g$ relative to that $h$ we need to do a union bound over all $h$'s.  (Note that every $g$ is 1/2-close to either the constant-1 function or the constant-0 function.)

%
%
%
%
}

\subsection{Proof of Lemma~\ref{lem:main-lower}}\label{sec:main-lower}


The following definitions will be useful.
Let $Y=(y^i:i\in [q])$ be a sequence of $q$ strings in $\{0,1\}^n$, $\alpha$ be a $q$-bit string,
  and $J\subset [n]$ be a set of variables of size $k$.
We say that $(Y,\alpha,J)$ is \emph{consistent} if 
\begin{equation}\label{consistency}
\alpha_i=\alpha_j\quad\text{for all $i,j\in [q]$ with\quad  $y^i_J=y^j_J$}.
\end{equation}
Given a consistent triple $(Y,\alpha,J)$, 
  we write $\rjunta_{ Y,\alpha,J}$ to denote the uniform distribution over all juntas $h$ 
  over $J$ that are consistent with $(Y,\alpha).$  More precisely, a draw of $\bh \leftarrow \rjunta_{Y,\alpha,J}$ is generated as follows: For each $z\in \{0,1\}^J$, if there exists a $y^{i}$ such that $y^{i}_J=z,$ then $\bh(x)$ is set to $\alpha_i$ for all $x\in \{0,1\}^n$ with $x_J=z$; 
  if no such $y^{i}$ exists, then a uniform random bit $\bb(z)$ is chosen independently and $\bh(x)$ is set to $\bb(z)$ for all $x$ with $x_J=z.$

To prove Lemma \ref{lem:main-lower}, we first derive from $A$ a new \emph{randomized} algorithm $A'$ that 
  works on triples $(\phi,\D,J)$ from the support of either $\YES^*$ or $\NO^*$.
Again for clarity we~use $\phi$ to denote a function from the 
  support of $\YES$/$\YES^*$ 
  or $\NO$/$\NO^*$, $f$ to denote a function from $\YES$/$\YES^*$ 
  and $g$ to denote a function from $\NO$/$\NO^*$.

In addition to being randomized, $A'$ differs from $A$ in two important ways:
\begin{flushleft}\begin{enumerate}
\item
Like $A$, $A'$  receives  
   samples $\bY\larr \D^q$ and $\phi(\bY)$, but 
unlike $A$, $A'$ also receives $J$ for free.\vspace{-0.06cm}
   
\item 
Unlike $A$, $A'$ does not make any black-box queries but 
  simply runs on the triple $(\bY,\phi(\bY),J)$ it receives at the beginning. 
So formally $A'$ is a randomized algorithm that runs on triples 
  $(Y,\alpha,J)$, where $Y=(y^i:i\in [q])$ is a sequence of $q$ strings, $\alpha$ is a $q$-bit string,
  and $J\subset [n]$ is a set of variables of size $k$,
  and outputs ``accept'' or ``reject.'' %
\end{enumerate}\end{flushleft}
A detailed description of the randomized algorithm $A'$ running on
  $(Y,\alpha,J)$ is as follows:
\begin{flushleft}\begin{enumerate}
\item First, if $(Y,\alpha,J)$ is not consistent, $A'$ immediately halts and rejects
(simply because this can never occur if $(Y,\alpha,J)$ is obtained from a triple $(f,\D,J)$ in the 
  support of $\YES^*$).
Otherwise $A'$ applies $A_1$ on $(Y,\alpha)$ to obtain a sequence $Z=(Z^i:i\in [q])$
  of $q$ strings.

\item Next, $A'$ draws $\bh' \leftarrow \rjunta_{Y,\alpha,J}$.
(This is the only part of $A'$ that is randomized.)

\item Finally, $A'$ runs $A_2(Y,\alpha,\bh'(Z))$ and outputs the same
  result (accept or reject).
\end{enumerate}\end{flushleft}
From the description of $A'$ above, whether it accepts a triple $(\phi,\D,J)$ or not
  depends on both the randomness of $\bY$ and $\bh'$. 
Formally we have
$$
\Pr\big[\text{$A'$ accepts $(\phi,\D,J)$}\hspace{0.03cm}\big]=
\Pr_{\bY,\bh'}\Big[\text{$(\bY,\balpha,J)$ is consistent and $A_2\big(\bY,\balpha,\bh'(\bZ)\big)=1$}\Big].
$$

Lemma~\ref{lem:main-lower} follows immediately from the following three lemmas
  (note that the marginal distribution of $(\boldf,\bD)$ in $\YES^*$ (or $(\bg,\bD)$ in $\NO^*$) 
    is the same as $\YES$ (or $\NO$)). 
In all three lemmas we assume that $A$ is a $q$-query non-adaptive deterministic
  algorithm while $A'$ is the~randomized algorithm derived from $A$ as described above.

\begin{lemma}[$A'$ behaves similarly on $\YES^*$ and $\NO^*$] \label{lem:1} We have
\ignore{
Let $T^\ast$ be any algorithm that is given access to a (function, distribution) pair $(\phi,\D_\phi)$ and a subset $J' \subset [n],
|J'|=k$ and makes $q\le 2^{k/3}/100$ calls to the sampling oracle and no calls to the query oracle.  Then}
\begin{align*}
\mathlarger{\Big|} \hspace{0.08cm}\E_{(\boldf,\bD,\bJ)\larr \YES^*}\Big[
  \Pr\big[\text{$A'$ accepts $(\boldf,\bD,\bJ)$}\big]\Big]
- \E_{(\bg,\bD,\bJ)\leftarrow \NO^*}\Big[\Pr\big[\text{$A'$ accepts $(\bg,\bD,\bJ)$} \big]\Big]
\hspace{0.03cm}\mathlarger{\Big|}\le 1/8.
\end{align*}
\end{lemma}

\begin{lemma}[$A$ and $A'$ behave identically on $\YES$ and $\YES^*$, respectively] \label{lem:2}
We have
\begin{equation}\label{eq:lem2}
\E_{(\boldf,\bD,\bJ)\larr \YES^*}\Big[  \Pr\big[\text{$A$ accepts $(\boldf,\bD)$}\big]\Big]
= \E_{(\boldf,\bD,\bJ)\larr \YES^*} \Big[\Pr\big[\text{$A'$ accepts $(\boldf,\bD,\bJ)$}\big]\Big].
\end{equation}
\end{lemma}

\begin{lemma}[$A$ and $A'$ behave similarly on $\NO$ and $\NO^*$, respectively] \label{lem:3}
We have
$$
\mathlarger{\Big|}\hspace{0.08cm}
  \E_{(\bg,\bD,\bJ )\larr \NO^*}\Big[\Pr\big[\text{$A$ accepts $(\bg,\bD)$}\big]\Big]
- \E_{(\bg,\bD,\bJ)\leftarrow\NO^*}\Big[\Pr\big[\text{$A'$ accepts $(\bg,\bD,\bJ)$}\big]\Big]
\hspace{0.03cm}\mathlarger{\Big|}\le 1/8.
$$
\end{lemma}

We start with the proof of Lemma \ref{lem:1}, which 
  says that a limited algorithm such as 
  $A'$ cannot effectively distinguish between a draw from $\YES^*$ versus $\NO^*$:


\begin{proof}[Proof of Lemma \ref{lem:1}]
Since $A'$  runs on $(Y,\alpha,J)$, it suffices to
  show that the distributions of $(\bY,\balpha,\bJ)$ induced from
  $\YES^*$ and $\NO^*$ have small total variation distance.
For this purpose we first note that the distributions of $(\bY,\bJ)$ induced from
  $\YES^*$ and $\NO^*$ are identical: In both cases, $\bY$ and $\bJ$ are independent;
  $\bJ$ is a random subset of $[n]$ of size $k$;
  $\bY$ is obtained by first sampling a subset $\bS$ of $\{0,1\}^n$ of size $m$
  and then drawing a sequence of $q$ strings from $\bS$ with replacement.

Fix a pair $(Y,J)$ in the support of $(\bY,\bJ)$.
We say $Y$ is \emph{scattered} by $J$ if $y^{i}_J\ne y^{j}_J$ for all $i\ne j\in [q]$.
In particular this implies 
  that no string appears more than once in $Y$.
The following claim, whose proof we defer, shows that 
  $\bY$ is scattered by $\bJ$ with high probability.

\begin{claim} \label{claim:good}
We have that $\bY$ is scattered by $\bJ$ with probability at least 
$1 - O(2^{- k/3}).$
\end{claim}

Fix any $(Y,J)$ in the support of $(\bY,\bJ)$ such that $Y$ is scattered by $J$.  
We claim that the distributions of $\balpha$ conditioning on $(\bY,\bJ)=(Y,J)$ in the
  $\YES^*$ case and the $\NO^*$ case are identical,~from which it follows 
  that the total variation distance between the distributions of $(\bY,\balpha,\bJ)$
  in the two cases is at most $O(2^{- k/3})\le 1/8$ when $k$ is sufficiently large.
Indeed $\balpha$ is uniform over strings of length $q$ in both cases.
This~is trivial for  $\NO^*$. For $\YES^*$ note that $\balpha$ is determined
  by the~random $k$-junta $\boldf\larr \rjunta_J$; the claim follows from the assumption that
  $Y$ is scattered by $J$.
\end{proof}

\begin{proof}[Proof of Claim \ref{claim:good}]
We fix $J$ and show that
  $\bY$ is scattered by $J$ with high probability.
As strings of $\bY$ are drawn one by one, the probability of 
  $\by^i$ colliding with one of the previous samples is at most $(i-1)/m\le q/m$.
By a union bound, all strings in $\bY$ are distinct with
  probability at least
$$
1-q\cdot (q/m)=1-q^2/m=1-O(2^{-k/3}).
$$
Conditioning on this event, $\bY=(\by^i)$ is distributed precisely as a uniform random 
  sequence from $\{0,1\}^n$ with no repetition and thus each pair $(\by^i,\by^j)$ is distributed
  uniformly over pairs~of~distinct strings in $\{0,1\}^n$.
As a result, we have
$$
\Pr\left[\hspace{0.03cm}\by^{i}_J=\by^{j}_J\hspace{0.05cm}\right]
=\frac{2^{n-k}-1}{2^n-1}\le \frac{1}{2^k}.
$$
By a union bound over ${q\choose 2}$ pairs we have that the 
  probability of $\bY$ being scattered by $J$ is at least
$$
\big(1-O(2^{-k/3})\big)\cdot\left(1-{q\choose 2}\cdot 2^{-k}\right)\ge 1-O(2^{-k/3}).
$$
This finishes the proof of the claim.
\end{proof}

Next we prove Lemma \ref{lem:2}.

\begin{proof}[Proof of Lemma \ref{lem:2}]
The first expectation in (\ref{eq:lem2}) is equal to the probability that 
$$
A_2\mathlarger{\big(}\bY,\balpha,\boldf\big(A_1(\bY,\balpha\big)\mathlarger{\big)}=1, 
$$
where $(\boldf,\bD,\bJ)\larr \YES^*$, $\bY\larr \bD^q$ and $\balpha=\boldf(\bY)$.
For the second expectation, since the triple on which we run $A'$ is always consistent,
  we can rewrite it as the probability that
$$
A_2\mathlarger{\big(}\bY,\balpha,\bh'\big(A_1(\bY,\balpha)\big)\mathlarger{\big)}=1,
$$
where $(\boldf,\bD,\bJ)\larr \YES^*$, $\bY\larr \bD^q$, $\balpha=\boldf(\bY)$ and $\bh'\larr \rjunta_{\bY,\balpha,\bJ}$.

To show that these two probabilities are equal, we first note that 
  the distributions of $(\bY,\balpha,\bJ)$ are identical.
Fixing any triple $(Y,\alpha,J)$ in the support of $(\bY,\balpha,\bJ)$, which must 
  be consistent, we claim that the distribution of $\boldf$ conditioning on
  $(\bY,\balpha,\bJ)=(Y,\alpha,J)$ is exactly $\rjunta_{Y,\alpha,J}$.~This~is because, for each $z\in \{0,1\}^J$, if $y^i_J=z$ for some $y^i$ in $Y$, then
  we have $\boldf(x)=\alpha_i$ for all strings~$x$ with $x_J=z$;
  otherwise, we have $\boldf(x)=\bb(z)$ for all $x$ with $x_J=z$, where $\bb(z)$ is
  an independent and uniform bit.
This is the same as how $\bh'\larr \rjunta_{Y,\alpha,J}$ is generated.
It follows directly~from this claim that the two probabilities are the same. 
This finishes the proof of the lemma.
\ignore{
Given a fixed triple $(Y,s,J)$, note that a function $f$
  drawn from $\YES$ conditioning on $Y,s$ and $J$ is
  distributed exactly as $\rjunta_{(Y,s,J)}$.
The lemma then follows.
}
\end{proof}

Finally we prove Lemma \ref{lem:3}, the most difficult among the three lemmas:

\begin{proof}[Proof of Lemma \ref{lem:3}]
Similar to the proof of Lemma \ref{lem:2}, the first expectation is the probability of
$$
A_2\mathlarger{\big(}\bY,\balpha,\bg\big(A_1(\bY,\balpha\big)\mathlarger{\big)}=1, 
$$
where $(\bg,\bD,\bJ)\larr \NO^*$ and $\balpha=\bg(\bY)$, while the second expectation is the probability of
$$
\text{$(\bY,\balpha,\bJ)$ is consistent and }
A_2\mathlarger{\big(}\bY,\balpha,\bh'\big(A_1(\bY,\balpha)\big)\mathlarger{\big)}=1,
$$
where $(\bg,\bD,\bJ)\larr \NO^*$, $\bY\larr \bD^q$, $\balpha=\bg(\bY)$, and $\bh'\larr \rjunta_{\bY,\balpha,\bJ}$.
We note that the distributions of $(\bY,\balpha,\bJ,\bD)$ in the two cases are identical.

The following definition is crucial.
We say a tuple $(Y,\alpha,J,\D)$ in the support of $(\bY,\balpha,\bJ,\bD)$ is \emph{good} if
  it satisfies the following three conditions ($S$ below is the support of $\D$):
\begin{flushleft}\begin{enumerate}
\item[]$E_0$: $Y$ is scattered by $J$.\vspace{-0.12cm}
\item[]$E_1$: Let $Z=A_2(Y,\alpha)$. Then every $z$ in $Z$ and every 
  $x$ in $S\setminus Y$ have $d(x,z)>0.4n$.
(In $S\setminus Y$ we abuse notation and use $Y$ as 
  a set that contains all strings in the sequence $Y$.)\vspace{-0.12cm}
\item[]$E_2$: If a string $z$ in $Z$ satisfies $z_J=y_J$ for some $y$ in $Y$, then
  we must have $d(y,z)\le 0.4n$.
\end{enumerate}\end{flushleft}
We delay the proof of the following claim to the end.

\begin{claim}\label{lastclaim}
We have that $(\bY,\balpha,\bJ,\bD)$ is good with probability at least $7/8$.
\end{claim}

Fix any good $(Y,\alpha,J,\D)$ in the support and let $Z=A_2(Y,\alpha)$.
We finish the proof by showing that the distribution of $\bg(Z)$, a binary string of length $q$,
 conditioning on 
  $(\bY,\balpha,\bJ,\bD)=(Y,\alpha,J,\D)$ is the same as that of $\bh'(Z)$ with $\bh'\larr \rjunta_{Y,\alpha,J}$.
This combined with Claim \ref{lastclaim} implies that the difference of the two probabilities has absolute value at most $1/8$.  

To see this is the case we partition strings of $Z$ into $Z_w$, where each $Z_w$ is a nonempty set 
  that contains all $z$ in $Z$ with $z_J=w\in \{0,1\}^J$.
For each $Z_w$, we consider the following two cases:
\begin{flushleft}\begin{enumerate}
\item If there exists no string $y$ in $Y$ with $y_J=w$, then by $E_1$ strings
  in $Z_w$ are all far from strings of $S$ (i.e., the support of $\D$) in this section and thus,
  $\bg(z)=\bb(w)$ for some independent and uniform bit $\bb(w)$, for all strings $z\in Z_w$.
\item If there exists a $y$ in $Y$ with $y_J=w$ (which must be unique by $E_0$), say $y^i$,
  then by $E_1$ and $E_2$ strings in $Z_w$ are all close to $y$ and far from other strings
  of $S$ in this section.
As a result, we have
  $\bg(z)=\alpha_i$ for all strings $z\in Z_w$.
\end{enumerate}\end{flushleft}
So the conditional distribution of $\bg(Z)$ is identical to that of 
  $\bh'(Z)$ with $\bh'\hspace{-0.05cm}\larr\hspace{-0.05cm}\rjunta_{Y,\alpha,J}$.
This finishes the proof of the lemma.
\end{proof}

\ignore{
To show that (\ref{eq:goal}) holds, we fix any good outcome $(Y,J')$ of $(\bY,\bJ')$, any outcome $s \in \{0,1\}^q$ of $\bg(Y)$, and any outcome $Z$ of $\bZ$, and we show below that
\begin{align}
&\left|
  \Pr_{(\bg,\D_\bg)\leftarrow \NO,\bJ' \leftarrow (\bg,\D_\bg),\bY \leftarrow (\D_\bg)^q}\big[T_2((Y,s),(Z,\bg(Z)) \ \text{accepts} \ | \ \bJ'=J, \bY=Y, \bg(Y)=s\big]\right. \nonumber\\
  & \left.
- \Pr_{(\bg,\D_\bg)\leftarrow\NO,
\bJ' \leftarrow (\bg,\D_\bg),
\bh' \leftarrow \rjunta_{(\bY,\bg(\bY)),\bJ'},\bY \leftarrow (\D_\bg)^q}
\big[T_2((Y,s),(Z,\bh'(Z)))\ \text{accepts}\ | \ \bJ'=J, \bY=Y, \bg(Y)=s\big]
\hspace{0.03cm}\right|\nonumber\\
& \le 0.04, \label{eq:goal2}
\end{align}
which implies (\ref{eq:goal}).  We thus fix $(Y,J'),s,Z$ as described above for the rest of the argument.

}

\ignore{
}

\begin{proof}[Proof of Claim \ref{lastclaim}]
We bound the probabilities of $(\bY,\balpha,\bJ,\bD)$ violating each of the three conditions 
  $E_0$, $E_1$ and $E_2$ and apply a union bound.
By Claim \ref{claim:good}, $E_0$ is violated with probability $O(2^{-k/3})$.

For $E_1$, we fix a pair $(Y,\alpha)$ in the support and let $\ell\le q$ be the
  number of distinct strings in $Y$ and $Z=A_2(Y,\alpha)$.
Conditioning on $\bY=Y$, $\bS\setminus \bY$ is a uniformly random 
  subset of $\{0,1\}^n\setminus Y$ of size $m-\ell$. 
Instead of working with $\bS\setminus \bY$, we let $\bT$ denote a 
  set obtained by making $m-\ell$ draws from $\{0,1\}^n$ uniformly at random (with replacements).

On the one hand, the total variation distance between $\bS\setminus \bY$
  and $\bT$ is exactly the probability that either (1) $\bT\cap Y$ is nonempty or (2)
  $|\bT|<m-\ell$.
By two union bounds, (1) happens with probability at most 
$(m-\ell)\cdot (\ell/2^n)\le mq/2^n$ and (2) happens with probability at most $(m/2^n)\cdot m$.
As a result, the total variation distance is at most
  $(mq+m^2)/2^n$.
On the other hand, the probability that one of the strings of $\bT$ has distance at most
  $0.4n$ with one of the strings of $Z$ is at most
$ 
mq\cdot \exp(-n/100)
$ 
by a Chernoff bound followed by a union bound.
Thus, the probability of violating $E_1$ is at most (using the assumption that $k\le n/200$)
$$
(mq+m^2)/2^n+mq\cdot \exp(-n/100)
=O\big(2^{-n/300}\hspace{0.05cm}\red{\ln n} \big).
$$


For $E_2$, we fix a pair $(Y,\alpha)$ in the support and let $Z=A_2(Y,\alpha)$.
Because $\bJ$ is independent~from $(\bY,\balpha)$, it remains a subset of $[n]$ of size $k$
  drawn uniformly at random.
For each pair  $(y,z)$ with $y$ from $Y$ and $z$ from $Z$
  that satisfy $d(y,z)>0.4n$,
  the probability of $y_\bJ=z_\bJ$ is at most 
$$
{\frac {{0.6n \choose k}}{{n \choose k}}} \le (0.6)^k.
$$
Since there are at most $q^2$ many such pairs, it follows from a union bound that 
  the probability of violating $E_2$ is at most
$
q^2\cdot (0.6)^k\le e^{-0.07k}.
$

Finally the lemma follows from a union bound when $k$ (and thus, $n$) is sufficiently large.
\end{proof}


\bibliographystyle{abbrv}
\bibliography{allrefs}

\end{document}